\documentclass[11pt]{article}
\usepackage{amsmath,amssymb,latexsym}
\usepackage{amsthm}
\usepackage{graphicx}
\usepackage{enumerate}
\usepackage[T1]{fontenc}
\usepackage{rotating,amssymb,amsmath,amsfonts,delarray}
\usepackage[ruled,vlined]{algorithm2e}

\usepackage{color}
\usepackage{amsfonts}
\usepackage{hyperref}
\usepackage{breakcites}
\usepackage{booktabs}
\usepackage{tikz}
\usepackage{pgfplots}
\usepackage{todonotes}
\pgfplotsset{compat=1.17}
\usepackage{bbm}
\newcommand{\id}{\mathbbm{1}}
\newcommand{\eqd}{\,{\buildrel \mathrm{d} \over =}\,}

\usepackage{float}
\usepackage{comment}

\usetikzlibrary[arrows.meta,bending]
\usetikzlibrary{positioning}
\usetikzlibrary{patterns}
\usepackage{subcaption}
\usepackage{multirow}

\hypersetup{
  colorlinks=true,
  citecolor=blue,
  linkcolor=blue,
  filecolor=magenta,      
  urlcolor=cyan,
}

\definecolor{colEM}{RGB}{39 217 125}
\definecolor{colHC}{RGB}{140 217 39}
\definecolor{colCBW}{RGB}{217 38 48}

\definecolor{DarkBlue}{RGB}{0,0,139}
\definecolor{LightBlue}{RGB}{163,206,230}
\definecolor{Teal}{RGB}{0,128,128}
\definecolor{LightGreen}{RGB}{144,238,144}
\definecolor{Yellow}{RGB}{255,255,0}
\definecolor{BrightOrange}{RGB}{255,165,0}

\newtheorem{definition}{Definition}
\newtheorem{lemma}{Lemma}
\newtheorem{proposition}{Proposition}
\usepackage{thmtools}
\newtheorem{example}{Example}
\newtheorem{assumption}{Assumption}
\newtheorem{theorem}{Theorem}
\newtheorem{corollary}{Corollary}
\newtheorem{remark}{Remark}

\newcommand{\diff}{\,\mathrm{d}}

\newcommand{\Smm}{S^{\left(\bigtriangledown, \bigtriangledown\right)}}
\newcommand{\Smp}{S^{\left(\bigtriangledown, \bigtriangleup\right)}}
\newcommand{\Sindep}{S^{\left(\perp, \perp\right)}}
\newcommand{\Sindepp}{S^{\left(\perp, \bigtriangleup\right)}}
\newcommand{\Sindepm}{S^{\left(\perp, \bigtriangledown\right)}}
\newcommand{\Spp}{S^{\left(\bigtriangleup, \bigtriangleup\right)}}

\newcommand{\Smpt}{\widehat{S}^{\left(\bigtriangledown, \bigtriangleup\right)}}
\newcommand{\Sindept}{\widehat{S}^{\left(\perp, \perp\right)}}
\newcommand{\Sppt}{\widehat{S}^{\left(\bigtriangleup, \bigtriangleup\right)}}

\makeatletter
\newcommand\newtag[2]{#1\def\@currentlabel{#1}\label{#2}}
\makeatother

\usepackage{setspace}
\usepackage[margin=1.00in]{geometry}

\title{Collective risk models with FGM dependence}

\usepackage{authblk}
\author[1]{Christopher Blier-Wong\thanks{Corresponding author, \href{mailto:christopher.blierwong@utoronto.ca}{christopher.blierwong@utoronto.ca}}}
\author[2]{Hélène Cossette}
\author[2]{Etienne Marceau}

\affil[1]{Department of Statistical Sciences, University of Toronto, Canada}
\affil[2]{École d'actuariat, Université Laval, Canada}

\date{\today}

\begin{document}

\maketitle
\begin{abstract}
We study copula-based collective risk models when the dependence structure is defined by a Farlie-Gumbel-Morgenstern (FGM) copula. By leveraging a one-to-one correspondence between the class of FGM copulas and multivariate symmetric Bernoulli distributions, we find convenient representations for the moments and Laplace-Stieltjes transform for the aggregate random variable defined from collective risk models with FGM dependence. We examine different components of this collective risk model, aiming to better understand the impact of the assumed dependence between a claim's frequency and severity. Relying on stochastic ordering, we analyze the impact of dependence on the aggregate claim amount random variable. Even if the FGM copula may only induce moderate dependence, we illustrate through numerical examples that the cumulative effect of FGM dependence can lead to substantial variations in key risk measures on aggregate random variables defined from collective risk models. 


\end{abstract}

\section{Introduction}\label{sec:introduction}

In actuarial science, we use a collective risk model (CRM) to define the aggregate claim amount of a portfolio or a non-life insurance contract as the sum of a random number of random variables (rv)
\begin{equation}\label{eq:DefS}
	S = \begin{cases}
		\sum_{j = 1}^{N}X_j, & N > 0\\
		0, & N = 0
	\end{cases},
\end{equation}
The rv $N$ is the claim count (frequency) rv, and $\{X_j\}_{j \geq 1}$ is the sequence of claim amount rvs. 

For this paper, we consider the aggregate claim amount rv $S$ defined in \eqref{eq:DefS} under the alternative representation
\begin{equation}\label{eq:DefSNo2}
	S = \sum_{j=1}^{\infty}X_{j}\times \mathbbm{1}_{\left\{N\geq j\right\}},
\end{equation}
which is used notably in the proof of Theorem 3.4.5, on page 159 of \cite{muller2002comparison}. 

While both representations in \eqref{eq:DefS} and \eqref{eq:DefSNo2} of $S$ lead to the same aggregate claim amount rv, they have different interpretations that may influence how one thinks of dependence within CRMs. In \eqref{eq:DefS}, the reader is inclined to interpret CRMs as observing the count $N$ first, then observing a vector of $N$ claim amounts conditionally on the value taken by $N$. Intuitively, the dependence structure in the latter representation is typically hierarchical; $N$ affects the claim amounts, but the claim amounts may not impact the claim number.

In \eqref{eq:DefSNo2}, 
we consider an infinite sequence of claim amount rvs and express $S$ as the sum of the first $N$ terms.
There is no longer a hierarchical interpretation since one studies the entire CRM simultaneously, such that the claim count affects claim amounts and, in turn, claim amounts affect the claim count, potentially leading to a more intricate form of dependence that may not be captured using the representation in \eqref{eq:DefS}.

The structural foundation of the \textit{classical} CRM is based on three assumptions.
\begin{assumption}
\label{ass:CRMClassical}
The three assumptions of the classical CRM are:
    \begin{enumerate}[(a)]
    \item\label{ass:SuiteID} $\{X_j\}_{j \geq 1}$ forms a sequence of identically distributed strictly positive rvs.
    \item\label{ass:SuiteIndep} $\{X_j\}_{j \geq 1}$ forms a sequence of independent rvs.
    \item\label{ass:SuiteEtFreqIndep} $\{X_j\}_{j \geq 1}$ and claim number rv $N$ are independent.
    \end{enumerate}    
\end{assumption}
Assumption \ref{ass:CRMClassical}\eqref{ass:SuiteID} means that the claim amounts have the same marginal distribution, Assumption \ref{ass:CRMClassical}\eqref{ass:SuiteIndep} implies the claim amounts are independent of each other, and Assumption \ref{ass:CRMClassical}\eqref{ass:SuiteEtFreqIndep} signifies that the amount of a claim does not depend on the number of claims $N$ and the number of claims does not depend on the claim amounts. Throughout the paper, we will denote the aggregate claim amount rv within the classical CRM as $\Sindep$. One notably derives the following simple expressions for the expectation, the variance, and the Laplace–Stieltjes transform (LST) of $\Sindep$:
\begin{align}
        E\left[\Sindep\right] 
        & = E[N] E[X];
        \label{eq:esperance-indep} \\
        Var\left(\Sindep\right)
        & 
        = E[N]Var(X) + Var(N) E[X]^2; \label{eq:var-indep} \\
        \mathcal{L}_{\Sindep}(t) 
        & = E\left[\mathrm{e}^{-t\Sindep}\right] = \mathcal{P}_N(\mathcal{L}_X(t)), 
        \quad t \geq 0, \notag
\end{align}
where $\mathcal{P}_N$ is the probability generating function (pgf) of $N$ and $\mathcal{L}_X$ is the LST of $X$. The expression in \eqref{eq:esperance-indep} dates back to \cite{wald1945some}, and it has been referred to as Wald's equation since then. Wald's equation holds under more general conditions than the classical CRM, and we will discuss this in the appendix. There are plenty of tools to identify the distribution of $\Sindep$ and compute risk measures for $\Sindep$; see, for example, \cite[Chapters 6 and 10]{panjer1992insurance}, \cite[Chapter 4]{rolski1999stochastic}, and \cite[Chapter 2]{wuthrich2022non}. Also, Assumption \ref{ass:CRMClassical} allows one to estimate the parameters for the claim counts and the claim amounts separately; see \cite{parodi2023pricing}, \cite{wuthrich2022non}, \cite{klugman2018loss}. Finally, as mentioned in \cite[Chapter 6]{parodi2023pricing}, the classical CRM is the scientific basis for non-life insurance pricing and risk modelling in reinsurance.

In practice, Assumptions \ref{ass:CRMClassical}\eqref{ass:SuiteIndep} and \ref{ass:CRMClassical}\eqref{ass:SuiteEtFreqIndep} are 
not always verified. For this reason, researchers have been interested in studying CRMs with full dependence, allowing for dependence relations between (i) the claim number and the claim amounts and (ii) each claim amount, which is the focus of the current paper. Keeping Assumption \ref{ass:CRMClassical}\eqref{ass:SuiteID}, let us denote by $\aleph(F_N, F_X) = \aleph$ the class of CRMs defined with the sequence $(N, \underline{X})$, whose components could be dependent. 

In the actuarial science literature, we count three common approaches to considering dependence between the components of a (single-period) CRM:
\begin{itemize}
    \item two-part models where the severity component depends on the frequency \cite{gschlossl2007spatial, erhardt2012modeling, shi2015dependent, garrido2016generalized};
	\item bivariate models where the dependence between the frequency random variable and the (average) claim size random variable is introduced through common mixtures \cite{belzunce2006variability} or copulas \cite{kramer2013total, czado2012mixed};
	\item CRMs with full dependence (see also micro-level CRMs) where one may specify a dependence structure on the entire sequence $(N, \underline{X})$. Under this approach, one typically assumes that the claim rvs are identically distributed. 
\end{itemize}

For details on the first two approaches, we refer the reader to \cite{frees2016multivariate} for a comprehensive review. The scope of the current paper is on the third approach; hence, we review its literature in more detail. A CRM with full dependence was first considered in \cite{liu2017collective} in the context of dependence uncertainty to identify the dependence structure for the sequence $(N, \underline{X}) \in \aleph$, which leads to the riskiest aggregate claim amount rv within the class of CRMs. The authors of \cite{cossette2019collective} study CRMs with full dependence directly, presenting a method to capture the dependence structure by specifying the joint cumulative distribution function (cdf) of the frequency rv and the sequence of claim rvs through a copula. Under the copula-based approach to construct CRMs, for any $k \in \mathbb{N}_1$, we have
\begin{equation}\label{eq:cdf-nx1x2x3}
	F_{N, X_1, \dots, X_k}\left(n, x_1, \dots, x_k\right) =C\left(F_N(n), F_X(x_1), \dots, F_X(x_k)\right),
	\text{ }\left( n,x_{1},...,x_{k}\right) \in \mathbb{N}_1 \times \mathbb{R}_{+}^{k},
\end{equation}
where $C$ is a $(k+1)$-variate copula. In \cite{shi2020regression}, the authors study copula-based CRMs with full dependence. However, they assume conditional independence of the claim amount rvs given the claim count such that they only require a single bivariate copula to construct the joint density function of $(N, X_1, \dots, X_k)$ for any $k \in \mathbb{N}_1$. Then, \cite{oh2021copulabased} examine other families of copulas, including the normal, student and vine copulas. Through their study of micro-level copulas, \cite{ahn2021copula} argue that the bivariate copula approach may be inappropriate in practice since (i) identifying a copula family that precisely captures the inherent dependence between the frequency random variable and the average severity has been challenging, (ii) the dependence structure of the estimated copula is not easily interpretable, and (iii) estimating copula parameters with mixed data presents significant computational challenges.
In contrast to the work of \cite{ahn2021copula, oh2021predictive, younahn2022simple}, our interest is not on estimating the parameters of the components of a given CRM on real data: we aim to find a proper dependent CRM construction allowing the investigation and comprehension of effects of dependence on these models (although we will briefly address estimation with observed data).

Empirical experiments for CRMs with dependence indicate significant evidence against Assumptions \ref{ass:CRMClassical}\eqref{ass:SuiteIndep} and \ref{ass:CRMClassical}\eqref{ass:SuiteEtFreqIndep}. In \cite[p.295]{czado2012mixed}, the authors notice that the copula correlation parameter between the frequency and the average claim size is $0.1366$. Also, the authors of \cite[p.218]{gschlossl2007spatial} observe that compared to a policyholder with one observed claim, the expected average claim size for an observation with two observed claims decreases by about $25~\%$. If three or four claims have been reported, the expected average claim size decreases by about $75~\%$ and $92~\%$, respectively. The authors of \cite[p.367]{shi2020regression} have found that the association parameter between the frequency and severity in the Gaussian copula is estimated to be $-0.278$ with a standard error of $0.022$. In \cite[p.18]{lee2022insurance}, the authors note that the claim size negatively correlates with the number of claims for all perils. 

The class of CRMs with dependence $\aleph$ being too complicated to study directly, we will introduce two assumptions on the dependence within the components of $(N, \underline{X})$. The first is on the pattern of the dependence, and the second is on the family of copulas used to construct a copula-based CRM. We will study these assumptions' impact separately and when considered together. 

The first assumption for CRMs with full dependence is common in the literature on CRMs and supposes that the claim amount rvs are exchangeable. We revise Assumption \ref{ass:CRMClassical} as follows.
\begin{assumption}
\label{ass:CRMwithDependence}
    For CRMs with full dependence and exchangeable claim amount rvs, we assume that for every $k \in \{2,3,\dots, \sup \{k : \gamma_N(k) > 0\}\}$, we have $(N,X_1,\dots,X_k) \eqd (N,X_{\pi(1)},\dots,X_{\pi(k)})$ for any permutation $\pi$ of $\{1,2,\dots,k\}$, where $\eqd$ means equality in distribution and $\gamma_N(k):=\Pr(N = k)$.
\end{assumption}
The class of CRMs under Assumption \ref{ass:CRMwithDependence} will be denoted by $\aleph^*$.
\begin{remark} \label{rem:Assumptions}
    Assumption \ref{ass:CRMwithDependence} implies that
    \begin{itemize}
        \item $\{X_j\}_{1 \leq j \leq \sup \{k : \gamma_N(k) > 0\}}$ forms a sequence of exchangeable rvs. \label{ass:SuiteExch}
        \item Letting $A_N$ denote the support of $N$, the following property holds:
        \begin{equation}
    \text{for every $n \in A_N$, the sequence $\{X_j\}_{j\geq 1}$ is conditionally exchangeable given $N = n$.}
    \label{ass:Exchangeability}
    \end{equation}
    \item While $\aleph^*$ contains useful CRMs for practical purposes, some useful dependence structures are lost. For instance, CRMs with an autoregressive dependence structure between the claim amounts are included in $\aleph$ but not in $\aleph^*$. 
    \end{itemize}
    
\end{remark}
A different assumption is that the dependence structure within the components of the CRM is a Farlie-Gumbel-Morgenstern (FGM) copula using the copula-based approach. Combining the properties of closed-form solutions to many problems and the flexibility of possible shapes of dependence makes the FGM copula a valuable tool in actuarial science and related fields since it lets one easily study the impact of altering the dependence structure in a model. We define the class of CRMs with full FGM dependence, denoted $\aleph^{FGM}$, as follows. 
\begin{assumption}
\label{ass:CRMwithFGM}
    The CRM with full FGM dependence assumes that for every $k \in \{1,2,\dots,\sup\{k: \gamma_N(k)>0\}\}$, the joint distribution of $(N,X_1,\dots,X_k)$ is defined with a $(k+1)$-variate FGM copula, that is, $C$ in \eqref{eq:cdf-nx1x2x3} is a FGM copula. 
\end{assumption}


In this paper, we make novel contributions to the classes $\aleph^*$ and $\aleph^{FGM}$. We also study in detail the class $\aleph^{FGM*} := \aleph^{*} \cap \aleph^{FGM}$. We provide, in Figure \ref{fig:structure}, a graphical overview of the contributions of each class of CRMs. In Section \ref{sect:Preliminaires}, we provide notation and results about the family of FGM copulas. 
Section \ref{sect:class-general} provides general results about CRMs with full dependence and exchangeable claim amount rvs, that is, $\aleph^*$. 
In Section \ref{sect:StochasticRepresCRM-FGM}, we introduce a key result, Theorem \ref{thm:stochastic-representation-s}, which lets one study the properties of CRMs within $\aleph^{FGM}$. 
In Section \ref{sect:ComponentsCRM}, we conduct a detailed analysis of the properties of the components of the CRMs with FGM dependence with exchangeable claim amount rvs; we then study aggregate claim amount rvs within these CRMs in Section \ref{sect:AggregateClaimCRM}. 
We provide new results on dependence properties of CRMs with FGM dependence in Section \ref{sec:dependence}.
Section \ref{sect:conclusion} reviews the findings from the current paper and the potential trails to explore. 

\begin{figure}[ht]
    \centering
    \begin{tikzpicture}[node distance=2cm]
	\node (A) at (0,0) [draw, rectangle, rounded corners, align = center, fill = blue!40] {\textbf{CRMs with full dependence ($\boldsymbol{\aleph}$)}};
	\node (B) at (-4,-2) [draw, rectangle, rounded corners, align = left, fill = red!40] {\textbf{Exchangeable claim amount rvs ($\boldsymbol{\aleph^{*}}$)} \\ 
		-- {\it Satisfies Assumption \ref{ass:CRMwithDependence}}\\
		-- General properties in Section \ref{sect:class-general}};
	\node (C) at (4, -2) [draw, rectangle, rounded corners, align = left, fill = yellow!40] {\textbf{FGM copula-based CRMs ($\boldsymbol{\aleph^{FGM}}$)} \\
		-- {\it Satisfies Assumption \ref{ass:CRMwithFGM}}\\
		-- Stochastic representation in Section \ref{sect:StochasticRepresCRM-FGM}\\
		-- Dependence properties in Section \ref{sec:dependence}
	};
	\node (D) at (0, -5)  [draw, rectangle, rounded corners, align = left, fill = orange!40] {$\boldsymbol{\aleph^{FGM*} = \aleph^{*} \cap \aleph^{FGM}}$\\
		-- {\it Satisfies Assumptions \ref{ass:CRMwithDependence} and \ref{ass:CRMwithFGM}}\\
		-- Properties of components in Section \ref{sect:ComponentsCRM}\\
		-- Aggregate claim amount rv $S$ in Section \ref{sect:AggregateClaimCRM}
	};
	\draw[-{Latex[length=2mm,width=2mm]}, line width=0.2mm] (A) -- (B);
	\draw[-{Latex[length=2mm,width=2mm]}, line width=0.2mm] (A) -- (C);
	\draw[-{Latex[length=2mm,width=2mm]}, line width=0.2mm] (B) -- (D);
        \draw[-{Latex[length=2mm,width=2mm]}, line width=0.2mm] (C) -- (D);
\end{tikzpicture}
    \caption{Graphical overview of the contributions in this paper.}
    \label{fig:structure}
\end{figure}



\section{Preliminaries}
\label{sect:Preliminaires}

In this section, we will provide some notation and some required background on FGM copulas and collective risk models. 

\subsection{Notation}

Denote the set of strictly positive integers $\{1, 2, 3, \dots\}$ as $\mathbb{N}_1$ and $\mathbb{N}_1 \cup \{0\}$ as $\mathbb{N}_0$. We use the notation $\underline{Z} = \{Z_i\}_{i \geq 1}$ for a sequence of identically distributed rvs with $Z_i \,{\buildrel \mathrm{d} \over =}\, Z$, for $i \in \mathbb{N}_1$. 
Further, we use the notation $\boldsymbol{Z} = (Z_1, \dots, Z_d)$ for a $d$-dimensional vector of rvs. For two independent and identically distributed (iid) rvs $Z_1$ and $Z_2$, we denote $Z_{[1]} = \min(Z_1, Z_2)$, $Z_{[2]} = \max(Z_1, Z_2)$ and $\mu_{Z_{[j]}}^{(m)} := E\left[Z_{[j]}^{m}\right]$ for $j \in \{1, 2\}$ and $m \in \mathbb{N}_1$.

In this paper, the rv $I$ denotes a symmetric Bernoulli rv with cdf $F_I(x) = 0.5 \times \mathbbm{1}_{[0, \infty)}(x) + 0.5 \times \mathbbm{1}_{[1, \infty)}(x), x \geq 0$, where $\mathbbm{1}_{A}(x) = 1$, if $x \in A$ and 0 otherwise. For a random vector $\boldsymbol{I}$, denote its probability mass function (pmf)  by $f_{\boldsymbol{I}}(i_1, \dots, i_d)$ for $(i_1, \dots, i_d) \in \{0, 1\}^d$. Operations on vectors such as $\mathbf{x} < \mathbf{y}$, $\mathbf{x} \mathbf{y}$ and $\mathbf{x} + \mathbf{y}$ are meant component-wise. 

For a positive discrete rv $N$, we note the cdf as $F_{N}(n) = \Pr(N \leq n)$, the pmf as $\gamma_N(n) := \Pr(N = n)$, for $n \in \mathbb{N}_0$, and the pgf as $\mathcal{P}_{N}(t) = E[t^N] = \sum_{n = 0}^{\infty} t^n \gamma_N(n)$, for $t \in [-1, 1]$. Let $A_N$ denote the support of $N$, that is, $A_N = \{k: \gamma_N(k) > 0\}$. For a positive claim amount rv $X$, we denote its cdf and survival function, respectively, as $F_{X}(x)$ and $\overline{F}_{X}(x)$, for $x > 0$ (for notational convenience, we ignore claim amount distributions with strictly positive probability masses at $0$). We denote the LST of $X$ as $\mathcal{L}_{X}(t) = E[\exp\{-tX\}]$, for $t \geq 0$.

We will use notions of stochastic orders throughout this paper. Denote the usual stochastic order by $\preceq_{st}$, the dispersive order by $\preceq_{disp}$, the (increasing) convex order by $\preceq_{cx} (\preceq_{icx})$ and the supermodular order by $\preceq_{sm}$. We include relevant definitions and properties in the appendix to make this paper self-contained. 

\subsection{FGM copulas}\label{ss:fgm}

We now provide a few general results surrounding the family of FGM copulas, including its classical formulation and a recent representation, enabling us to develop this paper's main results. The expression of a $d$-variate FGM copula is 
\begin{equation} \label{eq:fgm-copula}
	C\left(u_1,\dots,u_d\right) =\prod_{m=1}^d u_m \left( 1+\sum_{k=2}^{d}\sum_{1\leq j_{1}<\cdots <j_{k}\leq d}\theta_{j_{1}\ldots j_{k}}\overline{u}_{j_{1}}\overline{u}_{j_{2}}\ldots \overline{u}_{j_{k}}\right), \quad \boldsymbol{u}\in [0,1]^{d},
\end{equation}
where $\overline{u}_{j}=1-{u}_{j}$, $j \in \{1,\dots,d\}$. In \cite{cambanis1977properties}, the author shows that \eqref{eq:fgm-copula} is a copula if the parameters belong to the set
\begin{equation}\label{eq:constraints-general}
	\mathcal{T}_d = \left\{(\theta_{12}, \dots, \theta_{1\dots d}) \in \mathbb{R}^{2^d - d - 1} : 1+\sum_{k=2}^{d}\sum_{1\leq j_{1}<\dots <j_{k}\leq d}\theta_{j_{1}\dots j_{k}}\varepsilon _{j_{1}}\varepsilon _{j_{2}}\dots
	\varepsilon_{j_{k}}\geq 0\right\},  
\end{equation}
for $\{\varepsilon_{j_{1}},\varepsilon _{j_{2}},\dots, \varepsilon_{j_{k}}\} \in \{-1,1\}^d$. 
The class $\mathcal{C}_d^{FGM}$ of $d$-variate FGM copulas has $2^d - d - 1$ copula parameters, and the set $\mathcal{T}_d$ contains $2^d$ constraints. Due to the complexity of $\mathcal{T}_d$ and the exponential number of parameters, FGM copulas are scarcely used in high-dimensional copula modelling. Within the context of CRMs with full dependence, it is important to easily admit high-dimensional dependence structures since the length of the sequence $(N, \underline{X})$ can be countably infinite when $\sup\{k: \gamma_N(k)>0\} = \infty$. We will introduce alternative representations of the FGM copula to circumvent the complicated constraints on parameter admissibility in high dimensions. 

In this paper, we study a special class of CRMs when the dependence structure underlying the joint distribution of the components of $(N, \underline{X})$ is a FGM copula, as stated in Assumption \ref{ass:CRMwithFGM}.
A significant advantage of FGM copulas is that being quadratic in each marginal, one often obtains closed-form solutions to problems of interest in actuarial science and applied probability, including computing correlation coefficients or risk measures. See, for instance, various results in \cite{genest2007everything} or applications in actuarial science \cite{barges2009tvarbased, woo2013note}. Some criticize using FGM copulas in risk management because of their restrictive dependence scope. Indeed, the range of Spearman's rho for bivariate FGM copulas is limited to $[-1/3,1/3]$. However, real-data insurance studies show moderate dependence between claim counts and claim amounts; the same goes for the dependence between different claim amounts.
The estimated Spearman correlation coefficients for the case studies presented in the introduction fall within the admissible range for FGM copulas. Those observations suggest that CRMs constructed with FGM copulas may be appropriate for real-life applications. Despite the FGM copula permitting only moderate dependence, we will show through examples that the cumulative impact of this dependence can result in substantial variations in risk measures within the class of CRMs constructed with FGM copulas.

One property of FGM copulas is that they are very flexible and may capture negative dependence. Indeed, a $d$-variate FGM copula may control the correlation between each $k$-tuple of the random vector for $k \in \{2, \dots, d\}$. However, in the classical formulation of the FGM copula, one needs to specify each of the $2^d - d - 1$ dependence parameters while satisfying the $2^d$ constraints in \eqref{eq:constraints-general}, which can be tedious in practice. Further, one does not have a straightforward way to understand and interpret the global dependence structure by looking at the individual dependence parameters. 

The parametrization and interpretation issue has recently been resolved in \cite{blier2021stochastic, blier2023risk}, where the authors present a one-to-one correspondence between the class of $d$-variate  symmetric Bernoulli distributions and the class of $d$-variate FGM copulas. In particular, we have from Theorem 3.2 and Corollary 3.3 of \cite{blier2021stochastic} that if $C$ is an FGM copula with a set of dependence parameters $\boldsymbol{\theta}$, then $C$ is equivalently given by
$$C\left(\boldsymbol{u}\right) 
		=\sum_{\boldsymbol{i}\in \{0,1\}^d}f_{\boldsymbol{I}}(\boldsymbol{i}) \prod\limits_{m=1}^{d} u_m  \left( 1+(-1)^{i_{m}} \overline{u}_m\right), \quad \boldsymbol{u} \in [0,1]^{d},$$
  where
\begin{equation}\label{eq:theta-to-fi}
    f_{\boldsymbol{I}}(\boldsymbol{i}) = \frac{1}{2^d} \left(1 + \sum_{k=2}^{d}\sum_{1\leq j_{1}<\cdots <j_{k}\leq d} \left( -1\right) ^{i_{j_{1}}+\cdots+i_{j_{k}}} \theta_{j_{1}\ldots j_{k}}  \right), \quad \boldsymbol{i} \in \{0,1\}^d.
\end{equation}
The random vector $\boldsymbol{I}$ belongs to the class of multivariate symmetric Bernoulli random variables. In this case, the dependence structure of $\boldsymbol{U}$ is entirely determined by the dependence structure of $\boldsymbol{I}$, and one obtains a better understanding of the shape of dependence for a set of FGM copula parameters by examining the distribution of the underlying symmetric Bernoulli random vector. Further, one may construct subfamilies of FGM copulas in high dimensions, with corresponding copula parameters that will always satisfy 
the tedious constraints in \eqref{eq:constraints-general}.

\subsection{Collective risk models with dependence}
\label{sect:class-general}



Let us recall a few properties of CRMs within $\aleph^*$, that is, CRMs satisfying Assumption \ref{ass:CRMwithDependence}. 
The conditional exchangeability property from \eqref{ass:Exchangeability} implies that 
(i) the elements of the sequence $\{X_i\}_{i \geq 1}$ share the same distribution $F_X$; 
(ii) for any $n \in A_N$, the elements of the sequence $\{X_i \vert N = n\}_{1 \leq i \leq n}$ share the same distribution $F_{X \vert N = n}$. 
Therefore, throughout the paper, we will refer to arbitrary claim amounts as $X$ and arbitrary claim amounts given a claim number $n \in A_N$ as $X \vert N = n$, dropping the index. 

Conditioning on the claim number rv $N$ and from \eqref{ass:Exchangeability}, the expectation of $S$ is given by
\begin{align}
    E[S]  & = \sum_{n = 1}^{\infty} \gamma_N(n) \times  nE[X \vert N = n] = E[N E[X \vert N]], \label{eq:moment-s-general}
\end{align}
assuming that the expectations exist.



Using the law of total variance with \eqref{ass:Exchangeability}, we can represent the variance of the aggregate claim amount rv $S$ as the sum of three components;
\begin{equation}\label{eq:var-general}
    Var(S) = \underbrace{E[NVar(X_1 \vert N)] 
            + E[N(N-1)Cov(X_1,X_2 \vert N)]}_{E[Var(S \vert N)]} 
            + \underbrace{Var(N E[X_1 \vert N])}_{Var( E[S \vert N])},    
\end{equation}
assuming that the expectations exist.
In \eqref{eq:var-general}, 
the first component $C_{EVar}^{Var}(S) := E[NVar(X_1 \vert N)]$ accounts for the variability of a claim amount rv conditionally on the rv $N$, the second component $C_{ECov}^{Var}(S) := E[N(N-1)Cov(X_1, X_2 \vert N)]$ reflects the impact of the covariance between two claim amount rvs conditionally on the rv $N$, and the third component $C_{VarE}^{Var}(S) :
= Var(N E[X_1 \vert N])$ captures the interplay of the variability of $N$ and the expectation of the claim amount rv conditionally on the rv $N$. In the setting of the model proposed by \cite{shi2020regression}, where it is assumed that the claim amounts are conditionally independent, the component $C_{ECov}^{Var}(S)$ is equal to $0$. We will later see that $C_{ECov}^{Var}(S)$ can be large compared with the two other components. See Theorem \ref{thm:var-s}
for a more detailed analysis of \eqref{eq:var-general} within the specific family of CRMs with full dependence.

Evoking again \eqref{ass:Exchangeability}, the LST of $S$ is
\begin{equation}\label{eq:lst-s}
	\mathcal{L}_S(t) 
 = 
    \gamma_N(0) + \sum_{n = 1}^{\infty} \gamma_N(n) E[\mathrm{e}^{-t (X_1+\dots+X_n)} \mid N = n], \quad t \geq 0.
\end{equation}

Throughout this paper, we will use the expected value and the variance of $S$ to interpret the effect of dependence on the magnitude and variability of the aggregate random variable within collective risk models. On the other hand, the LST of $S$ will help us identify the distribution of $S$ and implement a numerical procedure to compute the pmf of $S$ when the claim amount rv $X$ has a positive discrete distribution.

\section{Stochastic representation in CRMs with FGM dependence}\label{sect:StochasticRepresCRM-FGM}

Following Assumption \ref{ass:CRMwithFGM}, recall that for any $k \in \mathbb{N}_1$, the joint cdf of $(N, X_1, \dots, X_k)$ is defined by
\begin{equation}\label{eq:c-modele-fgm}
	F_{N, X_1, \dots, X_k}\left(n, x_1, \dots, x_k\right) =C\left(F_N(n), F_X(x_1), \dots, F_X(x_k)\right),
	\text{ }\left( n,x_{1},...,x_{k}\right) \in \mathbb{N}_1 \times \mathbb{R}_{+}^{k},
\end{equation}
where $C$ is a $(k+1)$-variate FGM copula. Rapidly, one finds the latter construction inconvenient to derive results on aggregate claim amount rvs $S$ within collective risk models, so an alternative representation is required. 

\subsection{Main theorem}

In the following theorem, we state the main result of this paper, which provides a stochastic representation for the CRMs within $\aleph^{FGM}$. 
This stochastic representation allows us to consider any family of FGM dependence, and it is the stepping stone for all results derived in the remaining sections of this paper. 

\begin{theorem}
\label{thm:stochastic-representation-s}
Consider a CRM $(N, \underline{X}) \in \aleph^{FGM}$. 
\begin{enumerate}
	\item Then, for every $k \in \mathbb{N}_1$, the random vector $(N, X_1, \dots, X_k)$, whose joint cdf is defined in \eqref{eq:c-modele-fgm}, admits the representation
    \begin{equation} \label{eq:RepresentationNo1}
        (N, X_1, \dots, X_k) \,{\buildrel \mathrm{d} \over =}\, 
        \left((1-I_0)N_{[1]} + I_0 N_{[2]}, 
        (1-I_1)X_{[1], 1} + I_1X_{[2], 1}, 
        \dots, 
        (1-I_k)X_{[1], k} + I_kX_{[2], k}\right).
    \end{equation} 
	\item Let $S$ be the aggregate claim amount rv defined within the CRM $(N, \underline{X})$. The rv $S$ admits the representation
	\begin{equation}\label{eq:stochastic-s-fgm}
		S \,{\buildrel \mathrm{d} \over =}\, 
        \sum_{j = 1}^{\infty} \left((1-I_j)X_{[1], j} 
        + I_j X_{[2], j}\right) \times 
        \mathbbm{1}_{\{(1-I_0)N_{[1]} + I_0 N_{[2]} \geq j\}},
	\end{equation}
	or, equivalently,
		\begin{equation}\label{eq:stochastic-s-fgm2}
			S \,{\buildrel \mathrm{d} \over =}\, 
            \begin{cases}
				\sum_{j = 1}^{(1 - I_0)N_{[1]} + I_0 N_{[2]}} \left((1 - I_j)X_{[1], j} + I_j X_{[2], j}\right), 
                & (1 - I_0)N_{[1]} + I_0 N_{[2]} > 0 \\
				0, 
                & (1 - I_0)N_{[1]} + I_0 N_{[2]} = 0
					\end{cases}.
			\end{equation}
\end{enumerate}
\end{theorem}

\begin{proof}

Let us start with the first statement. For any $k\in \mathbb{N}_1$, let $\boldsymbol{U}$ be a $(k+1)$-variate random vector of uniform rvs such that $F_{\boldsymbol{U}} = C$. Then, from \cite[Lemma 2]{blier2023risk}, the components of $\boldsymbol{U}$ admit the stochastic representation given by $
	\boldsymbol{U} \,{\buildrel \mathrm{d} \over =}\, (\boldsymbol{1}-\boldsymbol{I}) \boldsymbol{U}_{[1]}
    + \boldsymbol{I} \boldsymbol{U}_{[2]}.$
where the pmf of $\boldsymbol{I}$ is given in \eqref{eq:theta-to-fi} and $\boldsymbol{1}$ is a $(k+1)$-variate vector of 1's. For a rv $Y$, we have that $Y_{[j]} \,{\buildrel \mathrm{d} \over =}\, F_{Y}^{-1}(U_{[j]})$ for $j \in  \{1, 2\}$ (see, for instance, \cite{scheffe1945nonparametric}). It follows that $(N, X_1, \dots, X_N)$ admits the representation
\begin{equation}\label{eq:stochastic-nx}
	((1 - I_0)N_{[1]} + I_0N_{[2]}, (1 - I_1)X_{[1], 1} + I_1X_{[2], 1}, \dots, (1 - I_k)X_{[1], k} + I_kX_{[2], k}).
\end{equation}

We now show part 2. Notice that since \eqref{eq:stochastic-nx} holds for every $k \in \mathbb{N}_1$, there exists a sequence of symmetric Bernoulli rvs $\{I_j\}_{\{j \geq 0\}}$ such that \eqref{eq:stochastic-nx} can be extended to $(N, \underline{X}) \in \aleph^{FGM}$. Then, \eqref{eq:stochastic-s-fgm} follows by substituting for every $k\in\mathbb{N}_1$ the representation from \eqref{eq:stochastic-nx} into \eqref{eq:DefS} and \eqref{eq:DefSNo2}.

\end{proof}

The stochastic representations in \eqref{eq:stochastic-s-fgm} and \eqref{eq:stochastic-s-fgm2} are analogues of \eqref{eq:DefS} or \eqref{eq:DefSNo2} respectively, but it has the great advantage of revealing the underlying dependence structure specific to the CRMs with FGM dependence. 
As noted in Section \ref{ss:fgm}, the dependence structure of $\boldsymbol{I}$ determines the dependence structure of the FGM copula. From the one-to-one correspondence between the family of FGM copulas and the family of multivariate symmetric Bernoulli distributions, it is equivalent to construct a sequence $(N, \underline{X}) \in \aleph^{FGM}$ from its cdf in \eqref{eq:cdf-nx1x2x3} or by defining a sequence $\{I_j\}_{j\geq 0}$ and using the representation in \eqref{eq:RepresentationNo1}.
In the remainder of the paper, we show the usefulness of Theorem \ref{thm:stochastic-representation-s}.

\subsection{Two special cases}

Let us introduce two dependence structures for $\{I_j\}_{j\geq 0}$ that will be useful throughout this paper. 


\begin{example}\label{ex:I-max-dep}
	Let $\{I_j\}_{j\geq 0}$ be a sequence of comonotonic rvs, that is, $I_j = I$ for all $j \geq 0$. It implies that, for every $k \in \mathbb{N}_1$, the value of the pmf of $(I_0,I_1,\dots,I_k)$ is non-zero in the two cases
	\begin{equation*}
		\Pr(I_0=0,I_1=0,\dots,I_k=0)
		=
		\Pr(I_0=1,I_1=1,\dots,I_k=1)
		=
		\frac{1}{2},
	\end{equation*}
	while it is equal to zero for all the other $2^{k+1}-2$ elements of the support of the distribution of $(I_0,I_1,\dots,I_k)$.
	For every $k \in \mathbb{N}_1$, the dependence structure of $(N,X_1, \dots, X_k)$ is the $(k+1)$-variate FGM copula called the extreme positive dependence (EPD) FGM copula in \cite{blier2021stochastic} and \cite{blier2023exchangeable}, with
	\begin{equation} \label{eq:epd}
		C\left(\boldsymbol{u}\right) = \prod_{j = 0}^{k} u_j \left(1 + \sum_{m = 1}^{\left\lfloor \frac{k+1}{2} \right\rfloor}\sum_{0\leq j_{1}<\cdots <j_{2m}\leq k} \overline{u}_{j_1}\cdots \overline{u}_{j_{2m}}\right), \quad \boldsymbol{u}=(u_0,u_1,\dots,u_k) \in [0,1]^{k+1},
	\end{equation}
	where $\lfloor y \rfloor$ is the floor function. See the two aforementioned references for a more detailed analysis of the EPD FGM copula. 
	We denote the aggregate claim amount rv within this CRM as $\Spp$.
	Rather than directly using the copula in \eqref{eq:epd} in the analysis of the corresponding aggregate claim amount rv $\Spp$, the representation in \eqref{eq:RepresentationNo1} 
	from Theorem \ref{thm:stochastic-representation-s} reveals that, for every $k \in \mathbb{N}_1$,  
	\begin{equation*}
		(N, X_1, \dots, X_k) \eqd
		\begin{cases}
			\left(N_{[1]}, X_{[1], 1}, \dots, X_{[1], k} \right),& I = 0 \\
			\left(N_{[2]}, X_{[2], 1}, \dots, X_{[2], k}\right),& I = 1
		\end{cases}.  
	\end{equation*}
	The expression in \eqref{eq:stochastic-s-fgm} in this specific case leads to the convenient expression
	\begin{equation} \label{eq:DefSexemple1}
		\Spp 
		\eqd
		\begin{cases}
			\sum_{j = 1}^{\infty} X_{[1], j}'\times \mathbbm{1}_{\{N_{[1]}'\geq j\}}, & I = 0 \\
			\sum_{j = 1}^{\infty} X_{[2], j}'\times \mathbbm{1}_{\{N_{[2]}'\geq j\}}, & I = 1
		\end{cases}.
	\end{equation}
	
\end{example}

For the CRM of the following example, we bring a slight but significant modification to the CRM with the EPD FGM copula.

\begin{example}\label{ex:I-max-dep-reverse}
	Let $\{I_j\}_{j\geq 0}$ be a sequence of rvs such that $I_0 = 1-I$ and $I_j = I$ for $j \in \mathbb{N}_1$. In other words, $\{I_j\}_{j \geq 1}$ forms a sequence of comonotonic rvs, while $I_0$ is countermonotonic with respect to each component of the sequence. 
	It implies that, for every $k \in \mathbb{N}_1$, the value of the pmf of $(I_0,I_1,\dots,I_k)$ is non-zero in two cases
	\begin{equation*}
		\Pr(I_0=1,I_1=0,\dots,I_k=0)
		=
		\Pr(I_0=0,I_1=1,\dots,I_k=1)
		=
		\frac{1}{2},
	\end{equation*}
	while it is equal to zero for all the other $2^{k+1}-2$ elements of the support of the distribution of $(I_0,I_1,\dots,I_k)$.
	For every $k \in \mathbb{N}_1$, the dependence structure of $(N,X_1, \dots, X_k)$ is the $(k+1)$~variate FGM copula  
	\begin{equation} \label{eq:epdno2}
		C\left(\boldsymbol{u}\right) = \prod_{j = 0}^{k} u_j 
		\left(1 
		-
		\sum_{m = 1}^{\left\lfloor \frac{k+1}{2} \right\rfloor}\sum_{1\leq j_{1}<\cdots <j_{2m-1}\leq k} \overline{u}_{0}\overline{u}_{j_1}\cdots \overline{u}_{j_{2m-1}}
		+
		\sum_{\substack{m = 1 \\ k \geq 2}}^{\left\lfloor \frac{k}{2} \right\rfloor} \sum_{1\leq j_{1}<\cdots <j_{2m}\leq k} \overline{u}_{j_1}\cdots \overline{u}_{j_{2m}}\right), 
	\end{equation} 
	for $\boldsymbol{u} \in [0,1]^{k+1}$.
	We denote the aggregate claim amount rv within this CRM as $\Smp$.
	Rather than directly using the copula in \eqref{eq:epdno2} in the analysis of the corresponding aggregate claim amount rv $\Smp$, the representation in \eqref{eq:RepresentationNo1} 
	from Theorem \ref{thm:stochastic-representation-s} reveals that, for every $k \in \mathbb{N}_1$, we have
	\begin{equation*}
		(N, X_1, \dots, X_k) \,{\buildrel \mathrm{d} \over =}\, 
		\begin{cases}
			\left(N_{[2]}, X_{[1], 1}, \dots, X_{[1], k} \right),& I = 0 \\
			\left(N_{[1]}, X_{[2], 1}, \dots, X_{[2], k} \right),& I = 1
		\end{cases}.
	\end{equation*}
	The expression in \eqref{eq:stochastic-s-fgm} of Theorem \ref{thm:stochastic-representation-s} becomes
	\begin{equation} \label{eq:DefSexemple2}
		\Smp 
		\,{\buildrel \mathrm{d} \over =}\,
		\begin{cases}
			\sum_{j = 1}^{\infty} X_{[1], j}'\times \mathbbm{1}_{\{N_{[2]}'\geq j\}}, & I = 0 \\
			\sum_{j = 1}^{\infty} X_{[2], j}'\times \mathbbm{1}_{\{N_{[1]}'\geq j\}}, & I = 1
		\end{cases}
		.
	\end{equation} 
\end{example}
In Example \ref{ex:I-max-dep}, the dependence structure between the claim frequency and severity is positive. On the other hand, the dependence structure in Example \ref{ex:I-max-dep-reverse} induces negative dependence between the claim frequency and severity. 
For every $k \in \{2,3,\dots\}$, the distributions of $(X_1,\dots,X_k)$ have the same dependence structure that corresponds to the $k-$variate EPD FGM copula. Notably, it implies that $\theta_{0k} = 1$ in the model described in Example \ref{ex:I-max-dep} while $\theta_{0k} = -1$ in the model described in Example \ref{ex:I-max-dep-reverse}, $k \in \mathbb{N}_1$ and $\theta_{k_1 k_2} = 1$ in the models of both Examples \ref{ex:I-max-dep} and \ref{ex:I-max-dep-reverse}, for $1 \leq k_1 < k_2 \leq \sup\{k : \gamma_N(k)> 0\}$.  
Based on the representation of $\Spp$ in \eqref{eq:DefSexemple1}, the aggregate claim amount rv is either a sum of a small number of claims of small amounts or a sum of a large number of claims of large amounts. 
Based on the representation of $\Smp$ in \eqref{eq:DefSexemple2}, the aggregate claim amount rv is either a sum of a large number of claims of small amounts or a sum of a small number of claims of large amounts.
The positive (or negative) dependence relation impacts the expectation of the aggregate claim amount rv, its variance (if it exists), and its cdf, as we will explore in the remainder of this paper. 


\subsection{Sampling algorithms}

We turn our attention to simulating from a CRM. Since we are usually interested in studying the aggregate rv within CRMs, it isn't necessary to sample the entire sequence $(N, \underline{X})$; we are only interested in observing $(N, X_1, \dots, X_N)$, where the length of the sequence depends on the first element of that sequence. 

Since we construct copula-based CRMs using \eqref{eq:DefSNo2}, we have the practical issue that we must simulate from $N$ simultaneously as simulating from $(X_1, \dots, X_N)$, without knowing the length of the latter vector. We solve this issue by conditioning on the rv $N$. 


Our simulation strategy requires two steps: first, simulate $N$, then simulate $(X_1, \dots, X_N)$ conditional on $N$. Since the cdf of $(X_1, \dots, X_n \vert N = n)$ will rarely be simple, one will usually have to resort to numerical optimization. The great advantage of the class $\aleph^{FGM}$ is the representation in \eqref{eq:RepresentationNo1}: one may simulate a sample of $N_{[2]}$, and one will be guaranteed that the number of claim amounts will be smaller or equal to $N_{[2]}$. We propose two approaches to simulate samples $(N, X_1, \dots, X_N)$ from a complete CRM $(N, \underline{X}) \in \aleph^{FGM}$. Each approach will have a particular use case. 

\begin{enumerate}
	\item One may simulate two iid samples of $N$, and denote the smallest one as an observation from $N_{[1]}$ and the largest one as an observation of $N_{[2]}$. Then, it suffices to simulate a sample of the random vector $(I_0, I_1, \dots, I_{N_{[2]}})$ and the random vector will be long enough to sample from the CRM using the representation in \eqref{eq:RepresentationNo1}, although some observations may remain unused. We detail this method in Algorithm \ref{algo:sample-stochastic-1}. 
	\item One may simulate a sample of $I_0$, then simulate a sample of ${N}$. Then, one may simulate a sample of the random vector $(I_1, \dots, I_{{N}} \vert I_0$) and the random vector will be the correct length to simulate from the CRM using the representation in \eqref{eq:RepresentationNo1}. We detail this procedure in Algorithm \ref{algo:sample-stochastic-2}. 
\end{enumerate}

\begin{algorithm}[ht]
	\KwIn{$f_{\boldsymbol{I}}$,$f_{N}$,$f_{X}$, number of simulations $n$}
	\KwOut{Samples of $(N, X_1, \dots, X_N)$}
	\nl \For{$l = 1, \dots, n$}{
		\nl Generate two iid rvs, $N_1^{(l)}$ and $N_2^{(l)}$, distributed as $N$\;
		\nl Set $N_{[1]}^{(l)} = \min\left(N_1^{(l)}, N_2^{(l)}\right)$ and $N_{[2]}^{(l)} = \max\left(N_1^{(l)}, N_2^{(l)}\right)$\;
		\nl Generate $\left(I_0^{(l)}, \dots, I^{(l)}_{N_{[2]}^{(l)}}\right)$\;
		\nl Set ${N}^{(l)} = \left(1 - I_{0}^{(l)}\right) N_{[1]}^{(l)} + I_0^{(l)} N_{[2]}^{(l)}$\;
		\nl	\If{$N^{(l)} = 0$}{
		\nl	Set $(N, \underline{X})^{(l)} = \left(N^{(l)}\right)$
		} 
		\nl \Else{
			\nl Generate two ${N}^{(l)}$-variate random vectors $\boldsymbol{Y}_1$ and $\boldsymbol{Y}_2$\;
			\nl Set $\boldsymbol{X}_{[1]}^{(l)}$ as pointwise minimum of $\boldsymbol{Y}_1$ and $\boldsymbol{Y}_2$\;
			\nl Set $\boldsymbol{X}_{[2]}^{(l)}$ as pointwise maximum of $\boldsymbol{Y}_1$ and $\boldsymbol{Y}_2$\;
			\nl Compute $X_j^{(l)} = \left(1 - I_{j}^{(l)}\right) X_{[1], j}^{(l)} + I_j^{(l)} X_{[2],j}$\;\label{algoline-x1}
                \nl Set $(N, \underline{X})^{(l)} = \left(N^{(l)}, X_1^{(l)}, \dots, X_{N^{(l)}}^{(l)}\right)$}
	}
	\nl Return $(N, \underline{X})^{(l)}$ for $l = 1, \dots, n$.
	\caption{Stochastic sampling method 1.} \label{algo:sample-stochastic-1}
\end{algorithm}
\begin{algorithm}[ht]
	\KwIn{$f_{\boldsymbol{I}}$,$f_{N}$,$f_{X}$, number of simulations $n$}
	\KwOut{Samples of $(N, X_1, \dots, X_N)$}
	\nl \For{$l = 1, \dots, n$}{
		\nl Generate one rv $I_0^{(l)} \eqd \mathrm{Bern}(0.5)$\;
		\nl Generate one rv ${N}^{(l)}$ from $N_{\left[1 + I_0^{(l)}\right]}^{(l)}$\;
		\nl	\If{$N^{(l)} = 0$}{
		\nl	Set $(N, \underline{X})^{(l)} = \left(N^{(l)}\right)$
		}
		\nl \Else{
			\nl Compute the pmf of $f_{I_1, \dots, I_{{N}^{(l)}} \vert I_0^{(l)}}(i_1, \dots, i_{{N}^{(l)}}) = 2f_{I_0, \dots, I_{{N}^{(l)}}}\left(I_0^{(l)}, i_1, \dots, i_{{N}^{(l)}}\right)$ for all $(i_1, \dots, i_{{N}^{(l)}}) \in \{0, 1\}^{{N}^{(l)}}$\;
			\nl Generate a random vector $(I_0^{(l)},\dots,I^{(l)}_{{N}^{(l)}})$ from $f_{I_1, \dots, I_{{N}^{(l)}} \vert I_0^{(l)}}$\;
                \nl Generate two ${N}^{(l)}$-variate random vectors $\boldsymbol{Y}_1$ and $\boldsymbol{Y}_2$\;
			\nl Set $\boldsymbol{X}_{[1]}^{(l)}$ as pointwise minimum of $\boldsymbol{Y}_1$ and $\boldsymbol{Y}_2$\;
			\nl Set $\boldsymbol{X}_{[2]}^{(l)}$ as pointwise maximum of $\boldsymbol{Y}_1$ and $\boldsymbol{Y}_2$\;
			\nl Compute $X_j^{(l)} = \left(1 - I_{j}^{(l)}\right) X_{[1], j}^{(l)} + I_j^{(l)} X_{[2],j}$ for $j = 1, \dots,N^{(l)}$\label{algoline-x2}\;
                \nl Set $(N, \underline{X})^{(l)} = \left(N^{(l)}, X_1^{(l)}, \dots, X_{N^{(l)}}^{(l)}\right)$}
	}
	\nl Return $(N, \underline{X})^{(l)}$ for $l = 1, \dots, n$.
	\caption{Stochastic sampling method 2.} \label{algo:sample-stochastic-2}
\end{algorithm}
One should use Algorithm \ref{algo:sample-stochastic-1} when one has access to an efficient sampling algorithm for the entire random vector $(I_0, I_1, \dots, I_k)$ for every $k \in \mathbb{N}_1$. 
On the other hand, if one defines the distribution of $\boldsymbol{I}$ using a pmf, the conditional pmf $\Pr(I_1 = i_1, \dots, I_d = i_d \vert I_0 = i_0)$ for $\boldsymbol{i}\in \{0, 1\}^{d+1}$ is simple to obtain and one should use Algorithm \ref{algo:sample-stochastic-2}. For the models in Examples \ref{ex:I-max-dep} and \ref{ex:I-max-dep-reverse}, both algorithms are efficient since it is simple to sample the entire random vector $\boldsymbol{I}$ at once. 

For both Algorithms \ref{algo:sample-stochastic-1} and \ref{algo:sample-stochastic-2}, there is no distributional hypothesis with regards to the frequency and severity distributions; one does not need to restrict one's analysis to distributions that are closed under order statistics. One disadvantage of this representation is that many observations are not used since only one of the samples from $X_{[1]}$ or $X_{[2]}$ is required in lines \ref{algoline-x1} of both Algorithms \ref{algo:sample-stochastic-1} and \ref{algo:sample-stochastic-2}. However, the author of \cite{baker2008OrderstatisticsbasedMethodConstructing} avoids this issue by storing unused samples from order statistics and using them for other samples; we also recommend this approach. 

    



\section{Analysis of the components of CRM with FGM dependence}
\label{sect:ComponentsCRM}

Aiming to understand the impact of the dependence construction between the frequency and severity of a claim, we investigate the behaviour of the conditional claim amount rv $X$ given the claim count rv $N$. Also, as mentioned in \cite[Section 4.3]{shi2020regression} and \cite{wuthrich2008stochastic}, one can use, for reserving purposes, the conditional distributions of the individual claim amount rv $X$ and of the aggregate claim amount rv $S$ given the number of claims $N$. 

In this section and the next, we will focus on CRMs satisfying Assumptions \ref{ass:CRMwithDependence} and \ref{ass:CRMwithFGM} and denote this class as 
$\aleph^{FGM*}$. Note that $\aleph^{FGM*} = \aleph^{*} \cap \aleph^{FGM}$ and that $(\aleph^{*} \cup \aleph^{FGM}) \setminus \aleph^{FGM*} \neq \emptyset$. Assumptions \ref{ass:CRMwithDependence} and \ref{ass:CRMwithFGM}, when considered together, implies that $\theta_{0k} = \theta_{01}$ and $\theta_{1k} = \theta_{12}$ for $k \in \mathbb{N}_1$. 
Further, $\theta_{0k_1k_2} = \theta_{012}$ and $\theta_{k_1k_2k_3} = \theta_{123}$ for $k_1, k_2, k_3 \in A_N$ and $k_1<k_2<k_3$. 
For this reason, we will use 
$\theta_{01}$ to refer to the bivariate dependence parameter associated with the dependence between $N$ and any claim amount; $\theta_{12}$ to refer to the bivariate dependence parameter between two claim amounts; 
$\theta_{012}$ to refer to the trivariate dependence parameter between $N$ and two claim amounts.
We provide below a general theorem useful to analyze the components of the CRMs with full FGM dependence and exchangeable claim amount rvs.

\begin{theorem}\label{thm:componentresult}
Consider a CRM $(N, \underline{X}) \in \aleph^{FGM*}$ and three functions $\varphi_0, \varphi_1$ and $\varphi_2$ such that all expectations $E[\varphi_0(N)]$, $E[\varphi_1(X)]$, and $E[\varphi_2(X)]$ exist, and assume that $E[\varphi_0(N)]>0$. 
Define the function $\Delta(F_Y; \phi)= \int_{\mathbb{R}}\phi(y) \diff F^2_Y(y) - \int_{\mathbb{R}}\phi(y) \diff (1 - \overline{F}^2_{Y}(y))$ for an integrable function $\phi$. Then, we have
\begin{align}
	E[\varphi_0(N)\varphi_1(X_1)\varphi_2(X_2)] &= \sum_{(i_0, i_1, i_2) \in \{0, 1\}^3} f_{I_0, I_1, I_2}\left(i_0, i_1, i_2\right) E\left[\varphi_0\left(N_{[1 + i_0]}\right)\right]E\left[\varphi_1\left(X_{[1 + i_1]}\right)\right]E\left[\varphi_2\left(X_{[1 + i_2]}\right)\right] \nonumber\\
 &= E\left[\varphi_0\left(N\right)\right]E\left[\varphi_1\left(X\right)\right]E\left[\varphi_2\left(X\right)\right] \nonumber\\
&\quad + \frac{\theta_{01}}{4}\Delta(F_N;\varphi_0) \left(E\left[\varphi_1\left(X_{[2]}\right)\right] E\left[\varphi_2\left(X_{[2]}\right)\right] - E\left[\varphi_1\left(X_{[1]}\right)\right] E\left[\varphi_2\left(X_{[1]}\right)\right]\right)\nonumber\\
	&\quad +\frac{\theta_{12}}{4} E\left[\varphi_0\left(N\right)\right] \Delta(F_X; \varphi_1)\Delta(F_X;\varphi_2) - \frac{\theta_{012}}{8} \Delta(F_N;\varphi_0)\Delta(F_X; \varphi_1)\Delta(F_X;\varphi_2)\label{eq:esperancefgh}
\end{align}
 and
\begin{align}
	E[\varphi_1(X_1)\varphi_2(X_2)\vert \varphi_0(N)] &=\frac{E[\varphi_0(N)\varphi_1(X_1)\varphi_2(X_2)]}{E[\varphi_0(N)]} \nonumber\\
	&= E\left[\varphi_1\left(X\right)\right]E\left[\varphi_2\left(X\right)\right] \nonumber\\
	&\quad + \frac{\theta_{01}}{4}\frac{\Delta(F_N;\varphi_0)}{E[\varphi_0(N)]} 
 \left(E\left[\varphi_1\left(X_{[2]}\right)\right] E\left[\varphi_2\left(X_{[2]}\right)\right] - E\left[\varphi_1\left(X_{[1]}\right)\right] E\left[\varphi_2\left(X_{[1]}\right)\right]\right)\nonumber\\
	&\quad +\frac{\theta_{12}}{4} \Delta(F_X; \varphi_1)\Delta(F_X;\varphi_2) - \frac{\theta_{012}}{8} \frac{\Delta(F_N;\varphi_0)}{E[\varphi_0(N)]}\Delta(F_X; \varphi_1)\Delta(F_X;\varphi_2).\label{eq:asdfasdf}
\end{align}

\end{theorem}

\begin{proof}
Using the stochastic representation given in \eqref{eq:RepresentationNo1} of Theorem \ref{thm:stochastic-representation-s} and by conditioning on $(I_0, I_1, I_2)$, we have
\begin{align*}
	E[\varphi_0(N)\varphi_1(X_1)\varphi_2(X_2)] &= \sum_{(i_0, i_1, i_2) \in \{0, 1\}^3} f_{I_0, I_1, I_2}\left(i_0, i_1, i_2\right) E\left[\varphi_0\left(N_{[1 + i_0]}\right)\right]E\left[\varphi_1\left(X_{[1 + i_1]}\right)\right]E\left[\varphi_2\left(X_{[1 + i_2]}\right)\right].
\end{align*}
Replacing $f_{I_0, I_1, I_2}\left(i_0, i_1, i_2\right)$ by its expression based on the natural parameters of the FGM copula given in Theorem 3.2 of \cite{blier2021stochastic}, that is, 
\begin{equation*}\label{eq:pmf-trivariate}
    f_{\boldsymbol{I}}(i_0, i_1, i_2) = \frac{1}{8}\left(1 + (-1)^{i_0 + i_1}\theta_{01} + (-1)^{i_0 + i_2}\theta_{02} + (-1)^{i_1 + i_2}\theta_{12} + (-1)^{i_0 + i_1 + i_2}\theta_{012}\right),
\end{equation*}
for $(i_0, i_1, i_2) \in \{0, 1\}^3$, and here with $\theta_{01}$ equal to $\theta_{02}$, the desired result follows after tedious algebra manipulations. The proof for the conditional expectation is similar.
\end{proof}

In the rest of this section, we use Theorem \ref{thm:componentresult} to analyze the components within $\aleph^{FGM*}$. The following corollary provides the expression of the cdf of $X$ given the value taken by the claim number rv.


\begin{corollary} \label{cor:conditional-cdf}
The conditional cdf of $(X \vert N=n)$, $n \in A_N$, is given by 
\begin{align*}
        F_{X \vert N = n}(x) 
        &= \frac{1}{\gamma_N(n)} 
        \sum_{(i_0,i_1)\in \{0,1\}^2} 
        f_{I_0,I_1}(i_0,i_1) \gamma_{N_{[1+i_0]}}(n) F_{X_{[1+i_1]}}(x),
        \quad x \geq 0.
\end{align*}
    or equivalently by
\begin{align}\label{eq:cond-cdf-x-n-v2}
        F_{X \vert N = n}(x) 
        &=
        F_X(x) 
        +
        \frac{\theta_{01}}{4}\frac{\gamma_{N_{[2]}}(n)-\gamma_{N_{[1]}}(n)}{\gamma_{N}(n)}\left(F_{X_{[2]}}(x)-F_{X_{[1]}}(x)\right), \quad x \geq 0. 
\end{align}
\end{corollary} 

\begin{proof}
The result directly follows from letting functions $\varphi_0,\varphi_1,\varphi_2$ in \eqref{eq:asdfasdf} be respectively defined as $\varphi_0(k)= \mathbbm{1}_{\{k=n\}}$, $\varphi_1(y)=\mathbbm{1}_{\{y\leq x\}}$, and $\varphi_2(y)=1$. 
\end{proof}

If we assume in \eqref{eq:cond-cdf-x-n-v2} that $X$ is a continuous rv, the pdf of $(X \vert N=n)$ is given by 
\begin{equation} \label{eq:fXcondN}
    f_{X \vert N=n} (x) =
    f_X(x) 
    + \frac{\theta_{01}}{2} 
    \frac{\gamma_{N_{[2]}}(n)-\gamma_{N_{[1]}}(n)}
    {\gamma_{N}(n)}\left(f_{X_{[2]}}(x)-f_{X_{[1]}}(x)\right),
        \quad x \geq 0, \quad n \in A_N.
\end{equation}
On the other hand, 
if $X$ is a discrete strictly positive rv, we find from \eqref{eq:cond-cdf-x-n-v2} that the conditional pmf of $(X \vert N=n)$ is 
\begin{equation*}
    f_{X \vert N = n}(k) =
    f_X(k)
    + \frac{\theta_{01}}{2} 
    \frac{\gamma_{N_{[2]}}(n)-\gamma_{N_{[1]}}(n)}
    {\gamma_{N}(n)}f_{X}(k)\left(2F_X(k-1) + f_X(k) - 1\right),
        \quad k \in \mathbb{N}_1, \quad n \in A_N.
\end{equation*}


\begin{remark}
    A general rule of thumb is that random vectors constructed using Sklar's theorem inherit the dependence properties of the original copula. That is not the case for CRMs with full dependence due to the discrete component, which affects every member of the sequence of claim rvs. In particular, if $(U_1,\dots, U_k \vert U_0)$ is a vector of conditionally independent rvs, then it is not necessarily the case that $(X_1,\dots X_k \vert N)$ is a vector of conditionally independent rvs. In fact, to the best of our knowledge, we have not been able to construct a CRM in $\aleph^{FGM}$ that induces conditional independence of the claim rvs given the count rv, other than the trivial case of the classical CRM found in Assumption \ref{ass:CRMClassical}. 
\end{remark}

\begin{remark}\label{rem:mediane}
Whenever $\varphi_0(k)= \mathbbm{1}_{\{k=n\}}$ in \eqref{eq:asdfasdf}, we have
\begin{equation}\label{eq:correction-n}
    \frac{\Delta(F_N; \varphi_0)}{E[\varphi_0(N)]} = \frac{\gamma_{N_{[2]}}(n)-\gamma_{N_{[1]}}(n)}{\gamma_{N}(n)} = 4\left(\frac{F_N(n-1)+F_N(n)}{2}-\frac{1}{2}\right), \quad n \in A_N,
\end{equation}
which is positive if $F_N(n - 1) \geq 0.5$, that is, when $n - 1$ is greater than the median of $N$. On the other hand, the term on the right-hand side of \eqref{eq:correction-n} is negative when $F_N(n) \leq 0.5$, that is, when $n$ is smaller than the median of $N$ (in the case where $F_N(n - 1) < 0.5 < F_N(n)$ the sign will depend on the average between $F_N(n - 1)$ and $F_N(n)$). Further, in Corollary \ref{cor:conditional-cdf}, where $\varphi_1(y) = \mathbbm{1}_{\{y \leq x\}}$, the impact on the dependence arising from $X_1$ is
\begin{equation}\label{eq:correction-x}
    \Delta(F_X; \varphi_1) = F_{X_{[2]}}(x)-F_{X_{[1]}}(x) = 2F_X(x)(F_X(x) - 1), \quad x \geq 0.
\end{equation}
We have $2F_X(x)(F_X(x) - 1) > 0$ if $F_X(x) > 0.5$, while 
$2F_X(x)(F_X(x) - 1) < 0$ if $F_X(x) < 0.5$. Therefore, we notice sign changes of \eqref{eq:correction-n} around the median of $N$ and of \eqref{eq:correction-x} at the median of $X$, which will lead to an increase or decrease of expectations of the form $E[\varphi_0(N)\varphi_1(X_1)h(X_2)]$ compared with the independence case. Such sign changes will impact stochastic ordering and help us analyze the components of the CRM later in this section. See also \cite[Section 2]{johnson1977generalized} for similar observations about the median and order statistics in copula regression. 

\end{remark}

The following example illustrates the relationship between the conditional pdf of $X$ given $N$, the dependence parameter $\theta_{01}$, and the median of $N$.


\begin{example} 
\label{ex:PdfOfXGivenN}
Let $X \eqd \mathrm{Gamma}(3, 1/10)$ such that $\mu_{X} = 30$
and
$N \eqd \mathrm{Geometric}(1/11)$ such that $\mu_{N} = 10$ and $F_{N}^{-1}(0.5) = 7$.
In Figure \ref{fig:CondPdfDeX}, we present the values of the conditional pdf of $X$ given the number of claims, obtained with \eqref{eq:fXcondN}, along with those of the marginal pdf when the dependence between the frequency and each claim amount is negative (with $\theta_{01} = -1$ as in Example \ref{ex:I-max-dep-reverse}) and positive (with $\theta_{01} = 1$ as in Example \ref{ex:I-max-dep}). 

From Remark \ref{rem:mediane}, one observes that when dependence is positive, we have that
\begin{itemize}
    \item the pdf of $X \vert N = n$ is to the left of the pdf of $X$ when $n$ is smaller than the median of $N$, meaning that $(X \vert N = n) \preceq_{st} X$;
    \item the pdf of $X \vert N = n$ is to the right of the pdf of $X$ when $n - 1$ is larger than the median of $N$, meaning that $X \preceq_{st} (X \vert N = n)$.
\end{itemize}
The stochastic orders are reversed when $\theta = -1$. Also from Remark \ref{rem:mediane}, we observe that $f_{X \vert N = n}(F_X^{-1}(0.5))$ are identical for all $n \in A_N$, that is, each conditional pdf cross at the point $F_X^{-1}(0.5)$. Note that \cite{shi2020regression} also compares conditional pdfs and obtains similar relationships but only for positive dependence parameters. 
\begin{figure}[ht]
    \centering
    \resizebox{0.48\textwidth}{!}{
\begin{tikzpicture}
	\begin{axis}[
		axis y line=left,
		axis x line=bottom,
		width = 4in, 
		height = 3in,
		ymin = 0,
		ymax = 0.05, 
		xmin = 0, 
		xmax = 100,
		scaled y ticks=false,
		yticklabel style={
			/pgf/number format/fixed,
			/pgf/number format/precision=2
		},
		xlabel={$x$},
		ylabel={pdf}, 
		title={$\theta_{01} = -1$},
		label style={font=\Large},
		tick label style={font=\Large},
		title style={font=\Large},
		]
		\addplot[black] table [x="x", y="marg", col sep=comma] {code/density_cond_neg_dep.csv};
		\addlegendentry{$f_X(x)$}
		
		\addplot[DarkBlue] table [x="x", y="1", col sep=comma] {code/density_cond_neg_dep.csv};
		\addlegendentry{$f_{X \vert N = 1}(x)$}
		
		\addplot[LightBlue] table [x="x", y="3", col sep=comma] {code/density_cond_neg_dep.csv};
		\addlegendentry{$f_{X \vert N = 3}(x)$}
		
		\addplot[Teal] table [x="x", y="5", col sep=comma] {code/density_cond_neg_dep.csv};
		\addlegendentry{$f_{X \vert N = 5}(x)$}
		
		\addplot[green] table [x="x", y="10", col sep=comma] {code/density_cond_neg_dep.csv};
		\addlegendentry{$f_{X \vert N = 10}(x)$}
		
		\addplot[orange] table [x="x", y="15", col sep=comma] {code/density_cond_neg_dep.csv};
		\addlegendentry{$f_{X \vert N = 15}(x)$}
		
		\addplot[red] table [x="x", y="30", col sep=comma] {code/density_cond_neg_dep.csv};
		\addlegendentry{$f_{X \vert N = 30}(x)$}
		
		
	\end{axis}
\end{tikzpicture}
}
\hfill
\resizebox{0.48\textwidth}{!}{
	\begin{tikzpicture}
		\begin{axis}[
			axis y line=left,
			axis x line=bottom,
			width = 4in, 
			height = 3in,
			ymin = 0,
			ymax = 0.05, 
			xmin = 0, 
			xmax = 100,
			scaled y ticks=false,
			yticklabel style={
				/pgf/number format/fixed,
				/pgf/number format/precision=2
			},
			xlabel={$x$},
			ylabel={pdf}, 
			title={$\theta_{01} = 1$},
						label style={font=\Large},
						tick label style={font=\Large},
						title style={font=\Large},
			]
			\addplot[black] table [x="x", y="marg", col sep=comma] {code/density_cond_pos_dep.csv};
			\addlegendentry{$f_X(x)$}
			
			\addplot[DarkBlue] table [x="x", y="1", col sep=comma] {code/density_cond_pos_dep.csv};
			\addlegendentry{$f_{X \vert N = 1}(x)$}
			
			\addplot[LightBlue] table [x="x", y="3", col sep=comma] {code/density_cond_pos_dep.csv};
			\addlegendentry{$f_{X \vert N = 3}(x)$}
			
			\addplot[Teal] table [x="x", y="5", col sep=comma] {code/density_cond_pos_dep.csv};
			\addlegendentry{$f_{X \vert N = 5}(x)$}
			
			\addplot[green] table [x="x", y="10", col sep=comma] {code/density_cond_pos_dep.csv};
			\addlegendentry{$f_{X \vert N = 10}(x)$}
			
			\addplot[orange] table [x="x", y="15", col sep=comma] {code/density_cond_pos_dep.csv};
			\addlegendentry{$f_{X \vert N = 15}(x)$}
			
			\addplot[red] table [x="x", y="30", col sep=comma] {code/density_cond_pos_dep.csv};
			\addlegendentry{$f_{X \vert N = 30}(x)$}
			
			
		\end{axis}
	\end{tikzpicture}
}
    \caption{Conditional pdf of $X$ given the frequency $N$.}
    \label{fig:CondPdfDeX}
\end{figure}
\end{example}

It is useful, both for modelling purposes and to understand the aggregate claim amount rv within a CRM with dependence, to study the conditional expectation of a claim amount given the number of claims.


\begin{corollary}
\label{corol:ConditionalEXgivenN}
The conditional expectation of $(X \vert N=n)$, $n \in A_N$, is 
\begin{align} 
E[X|N = n]&=\sum_{(i_0,i_1) \in \{0,1\}^2}f_{I_0, I_1}(i_0, i_1) \frac{\gamma_{N_{[1+i_0]}}(n)}{\gamma_N(n)}\mu_{X_{[1+i_1]}}, \quad n \in A_N \label{eq:EXgivenNequalkNo2}
\end{align}
or, equivalently,
\begin{align} 
E[X|N = n]
&=E[X] +\frac{\theta_{01}}{4}\frac{\left(\gamma_{N_{[2]}}(n)-\gamma_{N_{[1]}}(n)\right)}{\gamma_N(n)}\left(\mu_{X_{[2]}}-\mu_{X_{[1]}}\right). \label{eq:EXgivenNequalk}
\end{align}
\end{corollary}
\begin{proof}
Under the assumptions of Theorem \ref{thm:componentresult} and, letting the functions $\varphi_0,\varphi_1,\varphi_2$ be defined as $\varphi_0(k)=\mathbbm{1}_{\{k=n\}}$, $\varphi_1(x)=x$ and $\varphi_2(x)=1$ in \eqref{eq:asdfasdf}, the desired result follows.
\end{proof}

\begin{example}\label{ex:cond-exp} 
Fix some $\mu > 0$ and let $X \eqd \mathrm{Exp}(1/\mu)$ and $X' \eqd \mathrm{Pareto}(\alpha, \lambda)$, $\alpha > 1$ with $\lambda = \mu(\alpha - 1)$ such that $E[X] = E[X'] = \mu$. Then, $\mu_{X_{[2]}} - \mu_{X_{[1]}} = \mu$ and $\mu_{X'_{[2]}} - \mu_{X'_{[1]}} = \mu \times 2 \alpha / (2 \alpha -1)$.
Further, \eqref{eq:EXgivenNequalk}, becomes
 \begin{align*} 
    E[X | N = n] 
    &=
    \mu
    +\frac{\theta_{01}}{4}
    \frac{\gamma_{N_{[2]}}(n) - \gamma_{N_{[1]}}(n)}{\gamma_N(n)}\times \mu,
            \quad n \in A_N,
\end{align*}  
and 
\begin{align*} 
    E[X' | N = n] 
    &=
    \mu
    +\frac{\theta_{01}}{4}
    \frac{\gamma_{N_{[2]}}(n) - \gamma_{N_{[1]}}(n)}{\gamma_N(n)} \times \mu \times \frac{2 \alpha}{2 \alpha -1},
            \quad n \in A_N.
\end{align*}  
In Figure \ref{fig:cond-exp}, we depict the values 
of $E[X_1 | N = n]$ and $E[X'_1 | N = n]$, for $n \in \{0,1,\dots,100\}$, assuming $N \eqd \mathrm{NB}(4, 0.1)$, $\mu = 20$, and $\alpha = 1.5$, where $\mathrm{NB}$ stands for negative binomial distribution. 
It implies that  $F_{N}^{-1}(0.5) = 33$, $E[N] = 36$, $\mu_{X_{[1]}} = 10$, $\mu_{X_{[2]}} = 30$,
$\mu_{X'_{[1]}} = 5$, and $\mu_{X'_{[2]}} = 35$. 





\begin{figure}[ht]
\centering
\resizebox{\textwidth}{!}{
	\begin{tikzpicture}
		\begin{axis}[
			width = 7in, 
			height = 3.5in,
			ymin = 0,
			ymax = 40, 
			xmin = 0, 
			xmax = 100,
			ylabel shift=-20pt,
			xlabel={$k$},
			ylabel={$E[Z \vert N = k]$}, 
			extra x ticks={33},
			extra x tick labels={$F_{N}^{-1}(0.5)$},
			extra x tick style={ticklabel pos = upper},
			extra y ticks={5, 10, 20, 30, 35},
			extra y tick labels={$\mu_{X'_{[1]}}$, $\mu_{X_{[1]}}$, $\mu$, $\mu_{X_{[2]}}$, $\mu_{X'_{[2]}}$},			
                extra y tick style={ticklabel pos = right},
			legend style={at={(0.98,0.5)},anchor=east},
			label style={font=\Large},
			tick label style={font=\Large},
			title style={font=\Large},
			legend style={font=\Large},
			]
			\addplot[const plot, blue] table [x="V1", y="V2", col sep=comma] {code/cond_exp.csv};
			\addlegendentry{$Z = X$ and $\theta = 1$}

                \addplot[const plot, green] table [x="x", y="y1", col sep=comma] {code/cond_exp2.csv};
                \addlegendentry{$Z = X'$ and $\theta = 1$}

                \addplot[const plot, red] table [x="V1", y="V3", col sep=comma] {code/cond_exp.csv};
			\addlegendentry{$Z = X$ and $\theta = -1$}
                
			\addplot[const plot, orange] table [x="x", y="y2", col sep=comma] {code/cond_exp2.csv};
                \addlegendentry{$Z = X'$ and $\theta = -1$}

			\draw[color=black, dashed] 
			(axis cs:33, 0) -- (axis cs:33, 40);

                \draw[color=black, dotted] 
			(axis cs:0, 20) -- (axis cs:100, 20);
   
			\draw[color=black, dotted] 
			(axis cs:0, 30) -- (axis cs:100, 30);
			\draw[color=black, dotted] 
			(axis cs:0, 10) -- (axis cs:100, 10);

                \draw[color=black, dotted] 
			(axis cs:0, 35) -- (axis cs:100, 35);
			\draw[color=black, dotted] 
			(axis cs:0, 5) -- (axis cs:100, 5);
 

   
			
		\end{axis}
\end{tikzpicture}}
\caption{Conditional expectation in Example \ref{ex:cond-exp}.}
\label{fig:cond-exp}
\end{figure}

From Remark \ref{rem:mediane}, we have that the curves of $E[Z | N = n]$ cross at the median of $N$, which is 33 in this example. Also, note for $\theta = 1$ that $E[X \vert N = 0] = \mu_{X_{[1]}}$ while $\lim_{m \to \infty} E[X \vert N = m] = \mu_{X_{[2]}}$; the orders are reversed when $\theta = -1$. The conditional expectations reach $\mu_{X_{[1]}}$ and $\mu_{X_{[2]}}$, this is because the term on the right-hand side of (\ref{eq:correction-n}), for the current claim frequency distribution, is $-2$ for $n = 0$ or 2 for $n \to \infty$. It is, however, not guaranteed that the conditional expectations will reach $\mu_{X_{[1]}}$ or $\mu_{X_{[2]}}$. 
\end{example}

\begin{corollary}
\label{thm:ConditionalEX2givenN}

The conditional second moment of $(X \vert N=n)$, $n \in A_N$, is given by 
\begin{equation*} 
    E[X ^2 | N = n] 
    =  
    \sum_{(i_0,i_1) \in \{0,1\}^2}
    f_{I_0, I_1}(i_0, i_1) 
     \frac{\gamma_{N_{[1+i_0]}}(n)}{\gamma_N(n)}\mu^{(2)}_{X_{[1+i_1]}},
            \quad n \in A_N. \label{eq:EX2givenNequalkNo1}
\end{equation*}
or, equivalently,
\begin{equation*} 
    E[X^2 | N = n] 
    =  
    E[X^2] 
    +\frac{\theta_{01}}{4} 
    \frac{
    \left(\gamma_{N_{[2]}}(n)-\gamma_{N_{[1]}}(n)\right)}{\gamma_N(n)}\left(\mu_{X_{[2]}}^{(2)}-\mu_{X_{[1]}}^{(2)}\right),
            \quad n \in A_N. 
\end{equation*}
\end{corollary}

\begin{proof}
Under the assumptions of Theorem \ref{thm:componentresult} and, letting the functions $\varphi_0,\varphi_1,\varphi_2$ be defined as $\varphi_0(k)=\mathbbm{1}_{\{k=n\}}$, $\varphi_1(x)=x^2$ and $\varphi_2(x)=1$ in \eqref{eq:asdfasdf}, the desired result follows.
\end{proof}

\begin{corollary}\label{cor:correlations}  
Assume the expectations $E[N^k]$ and $E[X^k]$ exist for $k = 1, 2$. Then, 
\begin{enumerate}
    \item the covariance between $N$ and $X_1$ is given by
\begin{equation*} 
    Cov(N,X_1) 
    =  \frac{\theta_{01}}{4}
    (\mu_{N_{[2]}}-\mu_{N_{[1]}})
    (\mu_{X_{[2]}}-\mu_{X_{[1]}});
\end{equation*}
    \item the covariance between $X_1$ and $X_2$ is given by
\begin{equation*} 
    Cov(X_1,X_2) 
    =  \frac{\theta_{12}}{4}
    (\mu_{X_{[2]}}-\mu_{X_{[1]}})^2.
\end{equation*}
\end{enumerate}
\end{corollary}
\begin{proof}
The expectation of the product of $N$ and $X_1$ is obtained with Theorem \ref{thm:componentresult} with functions $\varphi_0,\varphi_1,\varphi_2$ defined as $\varphi_0(n)=n$, $\varphi_1(x)=x$ and $\varphi_2(x)=1$ in \eqref{eq:esperancefgh}. As for the expectation of the product of $X_1$ and $X_2$, we let functions $\varphi_0(n)=1$, $\varphi_1(x)=x$ and $\varphi_2(x)=x$ in \eqref{eq:esperancefgh}. Both expressions of the covariance directly follow.
\end{proof}


    


\begin{remark} 
    In Corollaries \ref{corol:ConditionalEXgivenN}, \ref{thm:ConditionalEX2givenN} and \ref{cor:correlations}, one observes that the (conditional) moments or covariances only depend on the difference between moments of $X_{[2]}$ and $X_{[1]}$ or of $N_{[2]}$ and $N_{[1]}$. To study the impact of the frequency or claim amount rvs on the components of the CRM, we will rely on the dispersive order. Note that if $Z \preceq_{disp} Z^\dagger$, then we have from Theorem 3.B.31 of \cite{shaked2007stochastic} that one may order the spacings according to the usual stochastic order, that is, $(Z_{[2]} - Z_{[1]}) \preceq_{st} (Z_{[2]}^\dagger - Z_{[1]}^\dagger)$. Therefore, for increasing functions $f, g$ or $h$ in Theorem \ref{thm:componentresult}, $N \preceq_{disp} N^\dagger$ and $X \preceq_{disp} X^\dagger$ implies $E\left[\varphi_0(N)\varphi_1(X_1)\varphi_2(X_2)\right] \leq E\left[\varphi_0(N^\dagger)\varphi_1(X_1^\dagger)\varphi_2(X_2^\dagger)\right]$.
\end{remark}

We conclude this section by providing the conditional covariance of a pair of claim amount rvs given the claim number rv.

\begin{lemma}
For a CRM $(N, \underline{X}) \in \aleph^{FGM*}$, we have
\begin{align*}
    Cov(X_1, X_2  \vert N = n)&= \frac{\theta_{12}}{4} 
	\left(\mu_{X_{[2]}}-\mu_{X_{[1]}}\right)^2 \nonumber\\
	& \quad -\frac{\theta_{012}}{8} 
	\frac{\left(\gamma_{N_{[2]}}(n)-\gamma_{N_{[1]}}(n)\right)}{\gamma_N(n)}\left(\mu_{X_{[2]}}-\mu_{X_{[1]}}\right)^2\\
    & \quad -\frac{\theta^2_{01}}{4} 
	\frac{\left(\gamma_{N_{[2]}}(n)-\gamma_{N_{[1]}}(n)\right)^2}{\gamma_N(n)^2}\left(\mu_{X_{[2]}}-\mu_{X_{[1]}}\right)^2, \quad n \in A_N.
\end{align*}
\end{lemma}
\begin{proof}
From Theorem \ref{thm:componentresult} with $\varphi_0(k) = \id_{\{k = n\}}$, $\varphi_1(x) = x$ and $\varphi_2(x) = x$, we have
\begin{align} 
    E[X_1 X_2 | N = n] 
    &=  
    E[X_1]^2 
    +\frac{\theta_{01}}{4} 
    \frac{\left(\gamma_{N_{[2]}}(n)-\gamma_{N_{[1]}}(n)\right)}{\gamma_N(n)}\left(\mu_{X_{[2]}}^2-\mu_{X_{[1]}}^2\right) +\frac{\theta_{12}}{4} 
    \left(\mu_{X_{[2]}}-\mu_{X_{[1]}}\right)^2 \nonumber\\
    &\qquad-\frac{\theta_{012}}{8} 
    \frac{\left(\gamma_{N_{[2]}}(n)-\gamma_{N_{[1]}}(n)\right)}{\gamma_N(n)}\left(\mu_{X_{[2]}}-\mu_{X_{[1]}}\right)^2
    ,\quad n \in A_N. \notag\label{eq:EX1X2givenNequalkNo2}
\end{align}
Using the definition 
$Cov(X_1, X_2 \vert N = n) = E[X_1X_2 \vert N = n] - E[X\vert N = n]^2$ and inserting \eqref{eq:EXgivenNequalk} completes the proof.
\end{proof}

\section{Analysis of aggregate claim amount rvs within CRMs}
\label{sect:AggregateClaimCRM}

In this section, we examine the distribution of the aggregate claim amount rv $S$ defined within a CRM $(N, \underline{X})\in \aleph^{FGM*}$, starting with its expectation and variance. Then, we derive the expression of its LST, and we show how to use the FFT algorithm to compute the probability masses of $S$ when the distribution of the claim amount rvs is discrete. Finally, we find exact methods to identify the distribution of the aggregate claim amount rv when the common distribution of the individual claim amount rvs belongs to the class of mixed Erlang distributions.

\subsection{Expectation of S}


The following theorem reveals the impact of the FGM dependence on the expectation of the aggregate claim amount rv $S$ within the family of CRMs within $\aleph^{FGM*}$. 

\begin{theorem}\label{thm:ES}
	Assume that $E[X]<\infty$ and $E[N] < \infty$. Then, the first moment of $S$ in \eqref{eq:stochastic-s-fgm} is
    	\begin{equation}\label{eq:expected-value-S-bern}
	E[S] 
    = \sum_{(i_0, i_1) \in \{0, 1\}^2}f_{I_0, I_1}(i_0, i_1) 
        \mu_{N_{[1 + i_0]}} \mu_{X_{[1 + i_1]}}, 
	    \end{equation} 
        or, equivalently,
        \begin{equation}
        \label{eq:expected-value-s-natural}
	E[S] 
        = E[S^{(\perp, \perp)}]
        + C_{Dep}^{Exp}(S),
	\end{equation}
where $E[S^{(\perp, \perp)}]$ is given 
in \eqref{eq:esperance-indep}
and 
\begin{equation}
\label{eq:ContributionEsperanceS}
    C_{Dep}^{Exp}(S) = \frac{\theta_{01}}{4} (\mu_{N_{[2]}}-\mu_{N_{[1]}}) (\mu_{X_{[2]}}-\mu_{X_{[1]}}), \quad \theta_{01} \in [-1,1],
\end{equation}
is the contribution of the dependence within the sequence $(N,\underline X)$ to $E[S]$. 
\end{theorem}

\begin{proof}
Inserting \eqref{eq:EXgivenNequalkNo2} from Corollary \ref{corol:ConditionalEXgivenN} in \eqref{eq:moment-s-general}, we have 
\begin{align*}
    E[S] 
    &=
    \sum_{n = 1}^{\infty} n \gamma_N(n) E[X \vert N = n] 
    \notag =
    \sum_{n = 1}^{\infty} n \gamma_N(n) 
    \sum_{(i_0,i_1) \in \{0,1\}^2} f_{I_0, I_1}(i_0, i_1) 
     \frac{\gamma_{N_{[1+i_0]}}(n)}{\gamma_N(n)}\mu_{X_{[1+i_1]}} 
    \notag \\
    &=
    \sum_{(i_0,i_1) \in \{0,1\}^2} f_{I_0, I_1}(i_0, i_1)
    \sum_{n = 1}^{\infty} n \gamma_{N_{[1+i_0]}}(n) \mu_{X_{[1+i_1]}}
    =
    \sum_{(i_0,i_1) \in \{0,1\}^2} f_{I_0, I_1}(i_0, i_1)
   \mu_{N_{[1+i_0]}}(n) \mu_{X_{[1+i_1]}},
\end{align*}
where the latter corresponds to the expression in \eqref{eq:expected-value-S-bern}.
Using \eqref{eq:EXgivenNequalk} from Corollary \ref{corol:ConditionalEXgivenN} in \eqref{eq:moment-s-general}, we have
\begin{align}
    E[S] 
    &= 
    \sum_{n = 1}^{\infty} n \gamma_N(n) E[X \vert N = n] 
    \notag \\
    & = 
    \sum_{n = 1}^{\infty} n \gamma_N(n) E[X] 
    + 
    \frac{\theta_{01}}{4} 
    \sum_{n = 1}^{\infty} n \gamma_N(n) 
    \frac{(\gamma_{N_{[2]}} - \gamma_{N_{[1]}}) (\mu_{X_{[2]}} - \mu_{X_{[1]}}) }{\gamma_N(n)}
    \notag \\
    & =
    E[N] E[X]
    + \frac{\theta_{01}}{4}  (\mu_{N_{[2]}} - \mu_{N_{[1]}}) (\mu_{X_{[2]}} - \mu_{X_{[1]}}),
    \notag\label{eq:PreuveES}
\end{align}
which is the desired result in \eqref{eq:expected-value-s-natural}.
\end{proof}

\begin{remark}
\label{rem:AnalysisES}
The expectation of $S$ in \eqref{eq:expected-value-s-natural} consists of the sum of two components. The first one is the expectation of the aggregate claim amount rv within the classical CRM. The second component, $C_{Dep}^{Exp}(S)$, reveals the contribution to the $E[S]$ of the positive or negative relation of dependence between the claim number rv $N$ and the sequence of claim amount rvs. The decomposition \eqref{eq:expected-value-s-natural} allows us to make the two following observations:
\begin{enumerate}
    \item \textbf{Impact of the dependence between the claim number rv and the sequence of the claim amount rvs}. 
    Since $(\mu_{N_{[2]}} - \mu_{N_{[1]}}) \geq 0$ 
    and $(\mu_{X_{[2]}} - \mu_{X_{[1]}}) \geq 0$ for any distribution of $N$ and $X$, a positive (negative) dependence relation between the claim number rv $N$ and the sequence of claim amount rvs induces a positive (negative) value to the dependence parameter $\theta_{01}$ leading to a positive (negative) value for the component $C_{Dep}^{Exp}(S)$.
    
    \item \textbf{Impact of the marginals}. 
    The marginal distributions of $N$ and $X$ have an impact on the second component $C_{Dep}^{Exp}(S)$, as laid out in Lemma \ref{lem:ImpactMarginals}.  
\end{enumerate}
\begin{lemma} \label{lem:ImpactMarginals}
Consider two CRMs $(N, \underline{X}) \in \aleph^{FGM*}(F_N, F_X)$ and $(N^{\dagger}, \underline{X}^{\dagger}) \in \aleph^{FGM*}(F_{N^\dagger}, F_{X^\dagger})$ such that the parameters $\theta_{01}$ of both CRMs coincide. If $N \preceq_{cx} N^{\dagger}$ 
and $X \preceq_{cx} X^{\dagger}$, 
then $C_{Dep}^{Exp}(S) \leq C_{Dep}^{Exp}(S^{\dagger})$.   
\end{lemma}
\begin{proof} 
We rewrite the expression $C_{Dep}^{Exp}(S)$ in 
\eqref{eq:ContributionEsperanceS} 
as follows
\begin{equation}
\label{eq:ContributionEsperanceSNo2}
    C_{Dep}^{Exp}(S) = \theta_{01} (\mu_{N_{[2]}}-\mu_{N}) (\mu_{X_{[2]}}-\mu_{X}), \quad \theta_{01} \in [-1,1].
\end{equation}
When $N \preceq_{cx} N^{\dagger}$ 
and $X \preceq_{cx} X^{\dagger}$,
it follows from Proposition 3.4.25(ii) of \cite{denuit2006actuarial} that 
\begin{equation}
\label{eq:JeVaisSouper}
\mu_{N_{[2]}} \leq \mu_{N_{[2]}^{\dagger}}, \quad
\mu_{X_{[2]}} \leq \mu_{X_{[2]}^{\dagger}}, \quad
\mu_{N} = \mu_{N^{\dagger}}, 
\quad \text{and} \quad
\mu_{X} = \mu_{X^{\dagger}}. 
\end{equation}
Combining \eqref{eq:ContributionEsperanceSNo2}
and \eqref{eq:JeVaisSouper}, 
we find that
$C_{Dep}^{Exp}(S) \leq C_{Dep}^{Exp}(S^{\dagger})$. 
\end{proof}

\end{remark}

\subsection{Variance of S}

In \eqref{eq:var-general}, we have provided the general formula for a CRM satisfying Assumption \ref{ass:CRMwithDependence}. We now provide an expression for the variance of the aggregate claim amount rv within a CRM $(N, \underline{X}) \in \aleph^{FGM*}$. 

\begin{theorem}
\label{thm:var-s}
Consider a CRM $(N, \underline{X}) \in \aleph^{FGM*}$. Assume $E[X^k] <  \infty$ and $E[N^k] < \infty$, for $k = 1,2$. The expressions of the three terms of the decomposition in \eqref{eq:var-general} of the variance of $S$ are given by
\begin{align}
	E[NVar(X \vert N)]
	&=E[N]Var(X) + \frac{\theta_{01}}{4}\left(\mu_{N_{[2]}} - \mu_{N_{[1]}}\right)\left(Var(X_{[2]}) - Var(X_{[1]})\right) \notag \\
	&\qquad\qquad - \frac{\theta_{01}^2}{16}\left(\mu_{X_{[2]}} - \mu_{X_{[1]}}\right)^2
		\sum_{n = 1}^{\infty} \gamma_{N}(n) n  \left(\frac{\gamma_{N_{[2]}}(n) - \gamma_{N_{[1]}}(n)}{\gamma_{N}(n)}\right)^2;
  \label{eq:VarSComponentNo1}
\end{align}
\begin{align}
	E[N(N-1)Cov(X_1, X_2\vert N)]
	&= \frac{\theta_{12}}{4}\left(\mu_{N}^{(2)} - \mu_{N}\right)\left(\mu_{X_{[2]}}-\mu_{X_{[1]}}\right)^2 \notag \\
	&\qquad\qquad -\frac{\theta_{012}}{8}\left(\mu_{N_{[2]}}^{(2)} - \mu_{N_{[2]}} - \mu_{N_{[1]}}^{(2)} + \mu_{N_{[1]}}\right)\left(\mu_{X_{[2]}}-\mu_{X_{[1]}}\right)^2 \notag \\
	&\qquad\qquad- \frac{\theta_{01}^2}{16}\left(\mu_{X_{[2]}}-\mu_{X_{[1]}}\right)^2 \sum_{n = 1}^{\infty}\gamma_{N}(n)n^2\left(\frac{\gamma_{N_{[2]}}(n)-\gamma_{N_{[1]}}(n)}{\gamma_N(n)}\right)^2 \notag \\
	&\qquad\qquad+ \frac{\theta_{01}^2}{16}\left(\mu_{X_{[2]}}-\mu_{X_{[1]}}\right)^2 \sum_{n = 1}^{\infty}\gamma_{N}(n)n\left(\frac{\gamma_{N_{[2]}}(n)-\gamma_{N_{[1]}}(n)}{\gamma_N(n)}\right)^2;
 \label{eq:VarSComponentNo2}
\end{align}
\begin{align}
	Var(NE[X \vert N])
	&= Var(N)E\left[X\right]^2 + \frac{\theta_{01}}{4}\left(Var\left(N_{[2]}\right) - Var\left(N_{[1]}\right)\right)\left(\mu_{X_{[2]}}^{2} - \mu_{X_{[1]}}^{2}\right) \notag \\
	&\qquad\qquad + \frac{\theta_{01}}{8}\left(\mu_{N_{[2]}}^2 - \mu_{N_{[1]}}^2\right)\left(\mu_{X_{[2]}}^2 - \mu_{X_{[1]}}^2\right) \notag\\
	&\qquad\qquad + \frac{\theta_{01}^2}{16}\left(\mu_{X_{[2]}} - \mu_{X_{[1]}}\right)^2\sum_{n = 1}^\infty \gamma_N(n)n^2\left(\frac{\gamma_{N_{[2]}}(n) - \gamma_{N_{[1]}}(n)}{\gamma_{N}(n)}\right)^2 \notag \\
	&\qquad\qquad - \frac{\theta_{01}^2}{16}\left(\mu_{N_{[2]}} - \mu_{N_{[1]}}\right)^2\left(\mu_{X_{[2]}} - \mu_{X_{[1]}}\right)^2.
 \label{eq:VarSComponentNo3}
\end{align}
The variance of $S$ also has the following decomposition: 
\begin{equation}
    Var(S) 
    = 
    Var\left(S^{(\perp,\perp)}\right) + C_{Dep}^{Var}(S),
    \label{eq:VarS}
\end{equation}
where $Var\left(S^{(\perp,\perp)}\right)$ is as given in \eqref{eq:var-indep} and 
\begin{align}
	C_{Dep}^{Var}(S) 
    &=  
    \frac{\theta_{01}}{4} \left[ \left(\mu_{N_{[2]}}           - \mu_{N_{[1]}}\right)           \left(Var\left(X_{[2]}\right) - Var\left(X_{[1]}\right)\right) +\right. \notag \\
	& \qquad\qquad           \left(Var\left(N_{[2]}\right) - Var\left(N_{[1]}\right)\right) \left(\mu_{X_{[2]}}^2         - \mu_{X_{[1]}}^2\right) +  \notag  \\
	& \qquad\qquad \left.    \frac{1}{2}\left(\mu_{N_{[2]}}^2         - \mu_{N_{[1]}}^2\right)     \left(\mu_{X_{[2]}}^2         - \mu_{X_{[1]}}^2\right) \right] \notag \\
	&\qquad - \frac{\theta_{01}^2}{16} \left(\mu_{N_{[2]}}-\mu_{N_{[1]}}\right)^2\left(\mu_{X_{[2]}}-\mu_{X_{[1]}}\right)^2 \notag \\
	& \qquad + \frac{\theta_{12}}{4}\left(\mu^{(2)}_{N}-\mu_{N}\right)\left(\mu_{X_{[2]}}-\mu_{X_{[1]}}\right)^2  \notag \\
	& \qquad 
	- \frac{\theta_{012}}{8}\left[\left(\mu^{(2)}_{N_{[2]}} - \mu_{N_{[2]}}\right) - \left(\mu^{(2)}_{N_{[1]}} - \mu_{N_{[1]}}\right)\right]\left(\mu_{X_{[2]}}-\mu_{X_{[1]}}\right)^2. 
 \label{eq:ContributionVarS}
\end{align} 
\end{theorem} 

\begin{proof}
    Applying Theorem \ref{thm:stochastic-representation-s} to each term of the decomposition of $Var(S)$ in \eqref{eq:var-general}, we find the components \eqref{eq:VarSComponentNo1}, \eqref{eq:VarSComponentNo2}, and
    \eqref{eq:VarSComponentNo3}.
    Replacing the three terms in Theorem \ref{thm:var-s} in \eqref{eq:var-general} and after rearrangements, 
    we find \eqref{eq:VarS} and \eqref{eq:ContributionVarS}. 
\end{proof}

It is natural to ask if any of the three contributions to the variance in Theorem \ref{thm:var-s} dominates the others. In the appendix supplement, we illustrate that this is not the case through a numerical example, where we see that the magnitude of each component depends on the distribution of claim amounts, claim counts and the dependence structure.

From Theorem \ref{thm:ES} and Theorem \ref{thm:var-s}, we derive expressions of the mean and variance of the aggregate claim amount rv with the specific dependence structures presented in Examples \ref{ex:I-max-dep} and \ref{ex:I-max-dep-reverse}. 

\begin{corollary}\label{cor:moments-mix-interpretation}
The first moment and the variance of $S^{\left( \bigtriangleup,\bigtriangleup\right)}$ are respectively given by
	$$E\left[S^{\left( \bigtriangleup,\bigtriangleup\right)}\right] = \frac{1}{2}\mu_{N_{[1]}}\mu_{X_{[1]}} + \frac{1}{2}\mu_{N_{[2]}}\mu_{X_{[2]}};$$
	\begin{align*}
		Var\left(\Spp\right) &= \frac{1}{2} \left\{\mu_{N_{[1]}}  Var\left(X_{[1]}\right) + Var\left(N_{[1]}\right)  \mu_{X_{[1]}}^2\right\}\\
		&\quad +\frac{1}{2} \left\{\mu_{N_{[2]}}  Var\left(X_{[2]}\right) + Var\left(N_{[2]}\right)  \mu_{X_{[2]}}^2\right\} \\
		& \quad +\left(\frac{\mu_{N_{[2]}}\mu_{X_{[2]}}-\mu_{N_{[1]}}\mu_{X_{[1]}}}{2}\right)^2.
	\end{align*}
 Similarly, the first moment and the variance of $S^{\left( \bigtriangledown,\bigtriangleup\right)}$ are 
	$$E\left[S^{\left( \bigtriangledown,\bigtriangleup\right)}\right] = \frac{1}{2}\mu_{N_{[1]}}\mu_{X_{[2]}} + \frac{1}{2}\mu_{N_{[2]}}\mu_{X_{[1]}};$$
	\begin{align*}
		Var\left(\Smp\right) &= \frac{1}{2} \left\{\mu_{N_{[1]}}  Var\left(X_{[2]}\right) + Var\left(N_{[1]}\right)  \mu_{X_{[2]}}^2\right\}\\
		&\quad +\frac{1}{2} \left\{\mu_{N_{[2]}}  Var\left(X_{[1]}\right) + Var\left(N_{[2]}\right)  \mu_{X_{[1]}}^2\right\} \\
		& \quad +\left(\frac{\mu_{N_{[2]}}\mu_{X_{[1]}} - \mu_{N_{[1]}}\mu_{X_{[2]}}}{2}\right)^2.
	\end{align*}
\end{corollary}

In \cite{blier2023risk}, the authors provide a closed-form formula for moments of the sum of Pareto or Weibull distributed rvs when the dependence structure underlying their joint distribution is a FGM copula. Similarly, we have exact expressions for the moments of aggregate claim amount rvs within CRMs when $(N, \underline{X}) \in \aleph^{FGM}$ and when the sequence of claim rvs $X$ are Pareto or Weibull distributed. We explore this statement in the following example. 
\begin{example}\label{ex:pareto}
Let $\underline{X}$ be a sequence of Pareto distributed rvs with survival function $\overline{F}_X(x) = \lambda^\alpha/(\lambda + x)^\alpha$ and mean $\lambda/(\alpha - 1)$, where $\lambda, \alpha > 0$ and $x\geq 0$. Then, $X_{[1]}$ is Pareto distributed with parameter $2\alpha$. It follows that $\mu_{X_{[1]}} = \lambda/(2\alpha - 1)$ and $\mu_{X_{[2]}} = \lambda(3\alpha-1)/(\alpha - 1)/(2\alpha - 1)$. Let $N$ be a rv following a geometric distribution with success probability $0 < p < 1$. We have $\gamma_{N}(k) = p(1 - p)^{k}$ for $k \in \mathbb{N}_0$ and $E[N] = (1-p)/p$. One may show that $N_{[1]}$ follows a geometric distribution with success probability $p(2-p)$. Then, for $\alpha > 1$, the expectations of $S^{\left(\bigtriangledown,\bigtriangleup\right)}$ and $S^{\left(\bigtriangleup,\bigtriangleup\right)}$ are
	$$E\left[S^{\left(\bigtriangledown,\bigtriangleup\right)}\right] = E[N]E[X] \times \left(\frac{1}{2}\frac{\alpha - 1}{2\alpha - 1}\frac{3 - p}{2 - p} + \frac{1}{2}\frac{3\alpha - 1}{2 \alpha - 1}\frac{1 - p}{2 - p} \right) = E[N]E[X]\left(1 - \frac{\alpha}{(2\alpha - 1)(2 - p)}\right);$$
	$$E\left[S^{\left( \bigtriangleup,\bigtriangleup\right)}\right] = E[N]E[X] \times \left(\frac{1}{2}\frac{\alpha - 1}{2\alpha - 1} \frac{1 - p}{2 - p} + \frac{1}{2}\frac{3\alpha - 1}{2 \alpha - 1} \frac{3 - p}{2 - p}\right) = E[N]E[X] \left(1 + \frac{\alpha}{(2\alpha - 1)(2 - p)}\right),$$
    while a similar (although omitted due to space constraints) expression for the variance is obtained for $\alpha > 2$ by using Theorem \ref{thm:var-s} or Corollary \ref{cor:moments-mix-interpretation}. 
 
    For a numerical application, we set $p = 10/11$, $\alpha = 2.1$ and $\lambda = 2200$, such that $E[N] = 0.1$ and $E[X] = 2000$. We present the mean and variance of $S^{\left(\bigtriangledown, \bigtriangleup\right)}$, $S^{\left(\perp, \perp\right)}$
	and $S^{\left(\bigtriangleup, \bigtriangleup\right)}$ in Table \ref{tab:pareto-mean-variance}. Within any CRM with FGM dependence, the FGM copula only induces mild positive and negative dependence between the claim number rv and the claim amount rvs. It also induces mild positive and negative dependence between two claim amount rvs. 
    However, the CRMs defined by different sequences $(N, \underline{X}) \in \aleph^{FGM}$ generate a wide range of possible values of the expectation and variance of the aggregate claim amount rvs.
	\begin{table}[ht]
		\centering
		\begin{tabular}{crr}
			& Mean     & Variance     \\ \hline
			$S^{\left( \bigtriangledown,\bigtriangleup\right)}$ & 79.6875  & 786~547   \\
			$S^{(\perp, \perp)}$                 & 200.0000 & 8~840~000  \\
			$S^{\left( \bigtriangleup,\bigtriangleup\right)}$  & 320.3125 & 16~133~409
		\end{tabular}
		\caption{Mean and variance of aggregate claim amount rvs defined within some CRMs with full FGM dependence in Example \ref{ex:pareto}.}\label{tab:pareto-mean-variance}
	\end{table}

\end{example}

\subsection{Laplace-Stieltjes transform of S}
We now turn our attention to cases where the aggregate claim amount rv derived from a CRM defined by a sequence $(N, \underline{X}) \in \aleph^{FGM*}$ will belong to a known family of distributions. Our approach will rely on identifying the LST of the aggregate claim amount rv derived from a CRM, which has a convenient form when using the stochastic formulation of the FGM copula. 

\begin{theorem}\label{thm:laplace-crm-identical}
	The LST of the aggregate claim amount rv in \eqref{eq:lst-s}, when $(N, \underline{X}) \in \aleph^{FGM*}$,  becomes
	\begin{equation} \label{eq:lse-with-k-No1}
	    \mathcal{L}_{S}(t) = \gamma_N(0) + \sum_{n = 1}^{\infty} E_{I_0, I_1, \dots, I_n}\left\{\gamma_{N_{[1 + I_0]}}(n) \prod_{j = 1}^{n} \mathcal{L}_{X_{[1 + I_j]}}(t) \right\}, \quad t \geq 0.
	\end{equation}
	Alternatively, denote $K_j = I_1 + \dots + I_j$ for all $j \in \mathbb{N}_1$. 
 Then, we have
\begin{equation}\label{eq:lse-with-k}
	\mathcal{L}_{S}(t) = \gamma_N(0) + \sum_{i \in \{0, 1\}} \sum_{n = 1}^{\infty}\gamma_{N_{[1 + i]}}(n) \sum_{k = 0}^{n}  \Pr(I_0 = i, K_n = k) \mathcal{L}_{X_{[1]}}(t)^{n-k} \mathcal{L}_{X_{[2]}}(t)^{k}, \quad t \geq 0.
\end{equation}
\end{theorem}

\begin{proof}
We begin by developing the expression of the conditional LST, 
    \begin{align}
        E\left[\mathrm{e}^{-t (X_1+\dots+X_n)} \mid N = n\right] = \frac{1}{\gamma_{N}(n)}E\left[\mathbbm{1}_{\{N = n\}}  \prod_{j = 1}^{k} \mathrm{e}^{-t X_j}  \right],
        \quad t \geq 0.
        \label{eq:MauditePreuveNo2}
    \end{align}
Then, using  the stochastic representation provided in Theorem \ref{thm:stochastic-representation-s} and conditioning on the symmetric Bernoulli rvs, \eqref{eq:MauditePreuveNo2} becomes 
    \begin{align}
    E\left[\mathrm{e}^{-t (X_1+\dots+X_n)} \mid N = n\right]
	&= \frac{1}{\gamma_{N}(n)}E_{I_0, I_1, \dots, I_k}\left[E\left[\mathbbm{1}_{\{N_{[1 + I_0]} = n\}}\right]\prod_{j = 1}^{k} E\left[\mathrm{e}^{-t X_{[1 + I_j]}} \right]\right]\nonumber\\
	&= \frac{1}{\gamma_{N}(n)}E_{I_0, I_1, \dots, I_k}\left[\gamma_{N_{[1 + I_0]}}\prod_{j = 1}^{k} \mathcal{L}_{X_{[1 + I_j]}}(t)\right],
    \quad t \geq 0. \label{eq:cond-lst}
    \end{align}
We obtain \eqref{eq:lse-with-k-No1}
by inserting \eqref{eq:cond-lst} into \eqref{eq:lst-s}. 
For \eqref{eq:lse-with-k}, we have 
$$\mathcal{L}_{S}(t) = \gamma_N(0) + \sum_{n = 1}^{\infty} E_{I_0,I_1,\dots,I_n}\left\{\gamma_{N_{[1 + I_0]}}(n) \mathcal{L}_{X_{[1]}}(t)^{I_1+\dots+I_n} \mathcal{L}_{X_{[2]}}(t)^{n -I_1-\dots-I_n}\right\},$$
 which, upon collecting like terms, becomes
	$$\mathcal{L}_{S}(t)= \gamma_N(0) + \sum_{n = 1}^{\infty} E_{I_0,K_n}\left\{\gamma_{N_{[1 + I_0]}}(n) \mathcal{L}_{X_{[1]}}(t)^{K_n} \mathcal{L}_{X_{[2]}}(t)^{n -K_n}\right\}, \quad t \geq 0.$$
 Evaluating the expectation and changing the order of summation leads to the result in \eqref{eq:lse-with-k}.
\end{proof}

The expression in \eqref{eq:lse-with-k} of Theorem \ref{thm:laplace-crm-identical} has two main advantages. First, it is practical for computational purposes since it combines common terms for permutations of $(i_1,\dots,i_n) \in \{0, 1\}^n$. We compute values of almost all numerical examples using \eqref{eq:lse-with-k}. We also use \eqref{eq:lse-with-k} of Theorem \ref{thm:laplace-crm-identical} to derive the results in Section \ref{ss:mxErl}.
Second, it is useful to define dependence structures for CRMs. It is sufficient to define a rv $K_n$ with support $\{0,1,\dots,n\}$ and mean $n/2$ for all $n \in \mathbb{N}_1$. An illustration of this second advantage in given in Corollary \ref{cor:lst-min-max}.

\begin{corollary}\label{cor:lst-min-max}
Assume that the conditions of Theorem \ref{thm:laplace-crm-identical} are satisfied. We consider the following three dependence structures:
\begin{itemize}
    \item If the sequence of rvs $(I_0, I_1, \dots)$ is defined as in Example \ref{ex:I-max-dep}, then the LST of the aggregate claim amount rv $\Spp$ is given by
    \begin{equation*}
        \mathcal{L}_{\Spp}(t) = \frac{1}{2}\mathcal{P}_{N_{[1]}}\left(\mathcal{L}_{X_{[1]}}(t)\right) + \frac{1}{2}\mathcal{P}_{N_{[2]}}\left(\mathcal{L}_{X_{[2]}}(t)\right), 
        \quad t \geq 0.    
    \end{equation*}

    \item If the sequence of rvs $(I_0, I_1, \dots)$ is defined as in Example \ref{ex:I-max-dep-reverse}, then the LST of the aggregate claim amount rv $\Smp$ is given by
    \begin{equation*}
        \mathcal{L}_{\Smp}(t) = \frac{1}{2}\mathcal{P}_{N_{[1]}}\left(\mathcal{L}_{X_{[2]}}(t)\right) 
        + \frac{1}{2}\mathcal{P}_{N_{[2]}}\left(\mathcal{L}_{X_{[1]}}(t)\right), 
        \quad t \geq 0.    
    \end{equation*}
    \item Let $\{I_j\}_{j \geq 1}$ be a sequence of comonotonic rvs and let $I_0$ be independent of the sequence $\{I_j\}_{j \geq 1}$. Let the aggregate claim amount rv within this CRM be defined as $\Sindepp$. Then, the LST of $\Sindepp$ is
    \begin{equation*}
        \mathcal{L}_{\Sindepp}(t) = \frac{1}{2}\mathcal{P}_{N}\left(\mathcal{L}_{X_{[1]}}(t)\right) + \frac{1}{2}\mathcal{P}_{N}\left(\mathcal{L}_{X_{[2]}}(t)\right), \quad t \geq 0.  
    \end{equation*}
\end{itemize}


\end{corollary}

We provide additional applications of Theorem \ref{thm:laplace-crm-identical} and examples in the appendix supplement.

\subsection{Mixed Erlang claims}\label{ss:mxErl}

One of the most important families of distributions in non-life actuarial science is the mixed Erlang family; practical applications of this class of distributions are provided in \cite{lee2010modeling}, \cite{albrecher2017}, and \cite{gui2018fitting}. This family of distributions has convenient properties when combined with FGM copulas; see \cite{blier2023risk} in the context of risk aggregation. We now extend these results to CRMs using Theorem \ref{thm:laplace-crm-identical}.

The authors of \cite{landriault2015note} show that the order statistics of iid mixed Erlang distributed rvs are also mixed Erlang distributed. If $X$ is mixed Erlang distributed with scale $\beta$ and sequence of masses $\underline{q}$, then $X_{[i]}$, $i \in \{1, 2\}$ is also mixed Erlang distributed with scale $2\beta$ and with the sequence of masses given by
\begin{equation*}
	q_{k, \{i\}} := \begin{cases}
			\frac{1}{2^{k-1}} \sum_{m = 0}^{k-1} \binom{k-1}{m} q_{m + 1} \left(1 - Q_{k-1-m}\right), & \text{for } i = 1\\
			\frac{1}{2^{k-1}} \sum_{m = 0}^{k-1} \binom{k-1}{m} q_{m + 1} Q_{k-1-m}, & \text{for } i = 2
		\end{cases},
\end{equation*}
where $Q_{k} = \sum_{m = 1}^k q_{m}$, for $k \in \mathbb{N}_1$ and $Q_{0} = 0$.
%

Let $\underline{X}$ form a sequence of claim amount rvs such that each claim amount rv is mixed Erlang distributed. We now show that the aggregate claim amount rv derived within a CRM defined by the sequence $(N, \underline{X}) \in \aleph^{FGM*}$ is mixed Erlang distributed. From Theorem \ref{thm:laplace-crm-identical}, we have
$$\mathcal{L}_{S}(t) = \gamma_N(0) + \sum_{i \in \{0, 1\}} \sum_{n = 1}^{\infty}\Pr(N_{[1 + i]} = n) \sum_{k = 0}^{n}  \Pr(I_0 = i, K_n = k) \mathcal{L}_{X_{[1]}}(t)^{n-k} \mathcal{L}_{X_{[2]}}(t)^{k},
\quad t \geq 0,$$
which becomes
\begin{align*}
	\mathcal{L}_{S}(t) &= \gamma_N(0) + \sum_{i \in \{0, 1\}} \sum_{n = 1}^{\infty}\Pr(N_{[1 + i]} = n) \sum_{k = 0}^{n}  \Pr(I_0 = i, K_n = k) \times \nonumber\\
	&\qquad \qquad \qquad \left\{\sum_{j = 1}^{\infty} q_{j, \{1\}} \left(\frac{2\beta}{2\beta + t}\right)^{j}\right\}^{n-k} \left\{\sum_{j = 1}^{\infty} q_{j, \{2\}} \left(\frac{2\beta}{2\beta + t}\right)^{j}\right\}^{k},
    \quad t \geq 0.
\end{align*}
Let $J_i$, $i \in \{1, 2\}$ be a rv with pmf $\Pr(J_i = j) = q_{j, \{1 + i\}}$, for $j \in \mathbb{N}_1$. Let $J_{k, n}$ be a rv with pgf
$$\mathcal{P}_{J_{k, n}}(z) = \mathcal{P}_{J_1}(z)^{n-k}\mathcal{P}_{J_2}(z)^{k}, \quad k \in \{0,1,\dots,n\}, \quad n \in \mathbb{N}_1, \quad z \in [-1,1].$$
Further, let $M$ be a rv with pgf
\begin{equation}\label{eq:pgf-skelette-mx-erl}
	\mathcal{P}_{M}(z) = \gamma_N(0) + \sum_{i \in \{0, 1\}} \sum_{n = 1}^{\infty}\Pr(N_{[1 + i]} = n) \sum_{k = 0}^{n}  \Pr(I_0 = i, K_n = k) \mathcal{P}_{J_{k, n}}(z),
    \quad z \in [-1,1].
\end{equation}
Then, the LST of $S$ becomes
\begin{equation}\label{eq:lst-mx-erl}
	\mathcal{L}_S(t) = \mathcal{P}_{M}\left(\frac{2\beta}{2\beta + t}\right), \quad t \geq 0.
\end{equation}
From the relation in \eqref{eq:lst-mx-erl}, it becomes simple to compute the cdf, VaR or TVaR for the aggregate claim amount rv derived within a CRM defined by a sequence $(N, \underline{X})\in \aleph^{FGM*}$. In particular, one may use numerical methods like the fast Fourier transform to compute the pmf of the rv $M$ from the pgf in \eqref{eq:pgf-skelette-mx-erl}. It suffices to manipulate the terms of the vector of probabilities associated with the pgf of $M$ to obtain the LST of $S$. The method described is exact if the rvs $N$, $J_1$, and $J_2$ have pmfs with finite support. 

\subsection{Discrete claim amount random variables}

The pgf of the aggregate claim amount rv $S$ will be convenient when the claim amount rvs are discrete. In this case, we adapt the LST of $S$ in \eqref{eq:lse-with-k} to obtain the pgf of $S$
\begin{equation*}
	\mathcal{P}_{S}(t) = \gamma_N(0) + \sum_{i \in \{0, 1\}} \sum_{n = 1}^{\infty}\Pr(N_{[1 + i]} = n) \sum_{k = 0}^{n}  \Pr(I_0 = i, K_n = k) \mathcal{P}_{X_{[1]}}(t)^{n-k} \mathcal{P}_{X_{[2]}}(t)^{k},
\end{equation*}
for $|t| \leq 1$. A discrete version of Corollary \ref{cor:lst-min-max} will also be useful for applications. Let $N$ be a rv with discrete support and $\underline{X}$ form a sequence of identically distributed rvs with discrete support. Then, the pgf of the aggregate claim amount rv within the CRMs are 
\begin{equation}\label{eq_pgf_s_pp}
	\mathcal{P}_{\Spp}(t) = \frac{1}{2}\mathcal{P}_{N_{[1]}}\left(\mathcal{P}_{X_{[1]}}(t)\right) + \frac{1}{2}\mathcal{P}_{N_{[2]}}\left(\mathcal{P}_{X_{[2]}}(t)\right)
\end{equation}
and
\begin{equation}\label{eq:pgf_s_mp}
	\mathcal{P}_{\Smp}(t) = \frac{1}{2}\mathcal{P}_{N_{[2]}}\left(\mathcal{P}_{X_{[1]}}(t)\right) + \frac{1}{2}\mathcal{P}_{N_{[1]}}\left(\mathcal{P}_{X_{[2]}}(t)\right),
\end{equation}
for $\vert t \vert \leq 1$. The expressions in \eqref{eq_pgf_s_pp} and \eqref{eq:pgf_s_mp} lead to an efficient fast Fourier transform method to compute the pmf of the rvs $\Spp$ and $\Smp$. In particular, when the distributions of $X_{[1]}$ and $X_{[2]}$ are unknown or have inconvenient forms, one may discretize the cdf of $X$ as in the following example. 

\begin{example}\label{ex:discretization}
Let $N \eqd \mathrm{NB}(10, 2/3)$. Further, let $X$ be log-normally distributed with $E[X] = 20$ and $Var(X) = 100$. We approximate the rv $X$ by a discrete rv $\widehat{X}$ using the moment-preserving method (see, for instance, \cite[Appendix E.2]{klugman2018loss}). Further, we approximate the aggregate claim amount rv $S$ within the CRM $(N, \underline{X}) \in \aleph^{FGM*}$ by the aggregate claim amount rv $\widehat{S}$ within the CRM $(N, \underline{\widehat{X}})\in\aleph^{FGM*}$. We study the rvs $\Smpt, \Sindept$ and $\Sppt$ that respectively approximate the rvs $\Smp, \Sindep$ and $\Spp$. Define the Value-at-Risk ($\mathrm{VaR}$) of a random variable $Z$ at confidence level $\kappa$ for $0<\kappa<1$, which corresponds to the left inverse of the cdf, that is, $\mathrm{VaR}_{\kappa}(Z) = F_Z^{-1}(\kappa)$. The Tail-Value-at-Risk ($\mathrm{TVaR}$) is defined as 
$$\mathrm{TVaR}_{\kappa}(Z) = \frac{1}{1-\kappa}\int_\kappa^{1} \mathrm{VaR}_{u}(Z) \diff u,$$
if $E[Z] < \infty$. We present risk measures for the discretized rvs $\Smpt, \Sindept$ and $\Sppt$ in Table \ref{tab:discretize-example}. In Figure \ref{fig:pmf-cdf-discretize}, we present the stairstep graphs of the pmf and cdf of the rvs $\Smpt, \Sindept$ and $\Sppt$. Even if FGM copulas induce moderate dependence, we observe that the risk measures can take very different values and that the tails of the cdfs in Figure \ref{fig:pmf-cdf-discretize} are noticeably different. 
\begin{table}[ht]
	\centering
	\begin{tabular}{rrrrr}
		          & $E[\widehat{S}]$ & $\sqrt{Var(\widehat{S})}$ & $\text{VaR}_{0.99}(\widehat{S})$ & $\text{TVaR}_{0.99}(\widehat{S})$ \\ \hline
		   $\Smpt$ &  92.08 &           47.20 &                 225 &                  260.21 \\
		$\Sindept$ & 100.00 &           59.17 &                 272 &                  314.03 \\
		   $\Sppt$ & 107.92 &           78.46 &                 336 &                  386.26
	\end{tabular}
	\caption{Values of expectation, standard deviation, VaR and TVaR of $\widehat{S}$ in Example \ref{ex:discretization}.}\label{tab:discretize-example}
\end{table}

\begin{figure}[ht]
	\centering
		\resizebox{0.48\textwidth}{!}{
	\begin{tikzpicture}
		\begin{axis}[
			width = 3in, 
			height = 3in,
			ymin = 0,
			xmin = 0, 
			ymax = 0.02, 
			xmax = 400,
			xlabel={$x$},
			ylabel={$\Pr\left(S = x\right)$}, 
			legend style={at={(1,1)},anchor=north east},
			]
			\addplot[const plot, green] table [x="x", y="pmfm", col sep=comma] {pmf_cdf_approx.csv};
			\addlegendentry{$\Pr(\Smpt = x)$}
			\addplot[const plot, blue] table [x="x", y="pmfi", col sep=comma] {pmf_cdf_approx.csv};
			\addlegendentry{$\Pr(\Sindept = x)$}
			\addplot[const plot, red] table [x="x", y="pmfp", col sep=comma] {pmf_cdf_approx.csv};
			\addlegendentry{$\Pr(\Sppt = x)$}			
		\end{axis}
\end{tikzpicture}}\hfill 
		\resizebox{0.48\textwidth}{!}{
	\begin{tikzpicture}
		\begin{axis}[
			width = 3in, 
			height = 3in,
			ymin = 0,
			xmin = 0, 
			ymax = 1, 
			xmax = 400,
			xlabel={$x$},
			ylabel={$\Pr\left(S \leq x\right)$}, 
			legend style={at={(1,0)},anchor=south east},
			]
			\addplot[const plot, green] table [x="x", y="cdfm", col sep=comma] {pmf_cdf_approx.csv};
			\addlegendentry{$\Pr(\Smpt \leq x)$}
			\addplot[const plot, blue] table [x="x", y="cdfi", col sep=comma] {pmf_cdf_approx.csv};
			\addlegendentry{$\Pr(\Sindept \leq x)$}
			\addplot[const plot, red] table [x="x", y="cdfp", col sep=comma] {pmf_cdf_approx.csv};
			\addlegendentry{$\Pr(\Sppt \leq x)$}			
		\end{axis}
\end{tikzpicture}}
	\caption{Pmf and cdf of $\Smpt$, $\Sindept$ and $\Sppt$ from Example \ref{ex:discretization}.}\label{fig:pmf-cdf-discretize}	
\end{figure}
\end{example}

\section{Dependence properties}\label{sec:dependence}

In this section, we build on the stochastic representation of Theorem \ref{thm:stochastic-representation-s} to analyze the impact of dependence on the aggregate claim amount rv $S$. We briefly recall the notions required for this section. Let $S$ and $S^{\dagger}$ be the aggregate claim amount rvs defined within the CRMs $(N, \underline{X})$ and $(N^{\dagger}, \underline{X}^{\dagger})$ respectively. 

\subsection{Impact of dependence and increasing convex order for aggregate claim amount rv}

As in \cite{cossette2019collective}, we aim to compare two aggregate claim amount rvs $S$ and $S^{\dagger}$ under the increasing convex order, allowing one to compare the rvs in terms of variability. To proceed, we need the following proposition, which comes from Proposition 4 of \cite{cossette2019collective} (itself an extension of Proposition \ref{prop:order-cx-s}, initially introduced by \cite{muller1997stop} and provided in Appendix \ref{app:defn-stochastic-orders}).

\begin{proposition}\label{prop:SuperModularIncreasingConvex}
    If $(N,X_1, \dots, X_k) \preceq_{sm} 
    (N^{\dagger},X_1^{\dagger}, \dots, X_k^{\dagger})$, 
    for every $k \in A_N$, 
    then $S \preceq_{icx} S^{\dagger}$. 
\end{proposition}

In the following theorem, we provide the conditions under which Proposition \ref{prop:SuperModularIncreasingConvex} holds for CRMs within $\aleph^{FGM}$. 
\begin{theorem}
\label{thm:Resultats}
If $ (I_0,I_1, \dots, I_k) 
        \preceq_{sm} 
        (I_0^{\dagger},I_1^{\dagger}, \dots, I_k^{\dagger})$, 
        for every $k \in \mathbb{N}_1$, then for every $k \in A_N$ 
        \begin{equation} \label{item:Theorem4point2}
            (N,X_1, \dots, X_k) 
        \preceq_{sm} 
        (N^{\dagger},X_1^{\dagger}, \dots, X_k^{\dagger}),
        \end{equation}
        and 
        \begin{equation}\label{item:Theorem4point2b}
            S \preceq_{icx} S^{\dagger}.
        \end{equation}
\end{theorem}

\begin{proof} 
The relation in \eqref{item:Theorem4point2} follows from Lemma 3 from \cite{blier2023risk}. 
The relation in \eqref{item:Theorem4point2b} results from combining \eqref{item:Theorem4point2} and Proposition \ref{prop:SuperModularIncreasingConvex}. 
\end{proof}

\begin{corollary}\label{cor: ExtremalElement}

For any aggregate claim amount rv within $(N, \underline{X}) \in \aleph^{FGM}$, we have
$S \preceq_{icx} \Spp$.
\end{corollary}
\begin{proof}
	For any $k \in \mathbb{N}_1$, let $(I_0^{\dagger}, I_1^{\dagger}, \dots, I_k^{\dagger})$ form a vector of comonotonic rvs. Then, for any random vector $(I_0, I_1, \dots, I_k)$, we have from \cite[Theorem 9.A.21]{shaked2007stochastic} that $(I_0, I_1, \dots, I_k) \preceq_{sm} (I_0^{\dagger}, I_1^{\dagger}, \dots, I_k^{\dagger})$. It follows from Theorem \ref{thm:Resultats} that $S^{\dagger}$ is larger under the increasing convex order than any other aggregate claim amount rv within a CRM $(N, \underline{X}) \in \aleph^{FGM}$. Noticing that $S^{\dagger} \eqd \Spp$ completes the proof. 
\end{proof}



\subsection{Negative dependence}

As stated in Section \ref{sec:introduction}, capturing negative dependence between $N$ and $\underline{X}$ is important since one observes this phenomenon on non-life insurance data. From Remark \ref{rem:AnalysisES}, we can conclude that the CRM $(N, \underline{X}) \in \aleph^{FGM*}$ that leads to the smallest expected value for the aggregate claim amount rv is obtained when $\theta_{0j} = -1$ for all $j \in \mathbb{N}_1$. This CRM corresponds to the model described in Example \ref{ex:I-max-dep-reverse}, that is, the CRM that leads to $\Smp$ (one may also show that $E[\Smp]$ is a lower bound for the expected value of the aggregate claim amount rv within any CRM $(N, \underline{X}) \in \aleph^{FGM}$). However, one has to be careful by concluding that $\Smp$ is a ``safe'' aggregate claim amount rv. Indeed, $\Smp$ is not the smallest (minimal) rv under the increasing convex order for aggregate claim amount rvs within CRMs $(N, \underline{X}) \in \aleph^{FGM*}$. We will investigate the effect of dependence on different dependence structures in the following example. 




\begin{example}\label{ex:finite-support}
    Consider a CRM where $\gamma_N(0) = \gamma_N(1) = 0.05$ and $\gamma_N(2) = 0.9$.  
    Further assume that $X_1 \eqd X_2 \eqd \mathrm{Gamma}(4, 1/100)$.
    In this example, we investigate the effect of dependence for eight CRMs within $\aleph^{FGM*}$. 
    Model \ref{modelindep} corresponds to the classical CRM, while Models \ref{modelsmp} and \ref{modelspp} are described in Examples \ref{ex:I-max-dep} and \ref{ex:I-max-dep-reverse}, respectively. 
    The dependence structure of Model \ref{modelsmm} is the extreme negative dependence in the case of exchangeable FGM copulas provided in Theorem 7 of \cite{blier2023exchangeable}. 
    In Models \ref{modelindepm} and \ref{model6}, the underlying Bernoulli rv $I_0$ are independent of $I_1$ and $I_2$, but $I_1$ and $I_2$ are respectively countermonotonic and comonotonic.     
    Models \ref{model4} and \ref{model7} are defined with the trivariate FGM copula with 2-independence and 3-dependence, studied in \cite[Section 7.2]{blier2021stochastic}.
    
    Table \ref{tab:ex-extremal-points} presents the dependence parameters of Models \ref{modelsmp} to \ref{modelspp} with the first two moments and the TVaR at level $\kappa=0.99$ of the aggregate claim amount rv $S$. We sort the rows by the values of the expectation of the aggregate claim amount rv. 
    We can apply Theorem \ref{thm:Resultats} to compare the aggregate claim amount rvs 
    \begin{equation}
    \label{eq:InegalitesSSSS}
        \Smm \preceq_{icx} \Sindep 
        \preceq_{icx} \Sindepp \preceq_{icx} \Spp
    \end{equation}
    and 
    \begin{equation*}
         \Sindepm \preceq_{icx} \Sindep 
        \preceq_{icx} \Sindepp \preceq_{icx} \Spp.
    \end{equation*}
    Therefore, it is unsurprising that values of the $\mathrm{TVaR}$ of the aggregate claim amount rvs defined within Models \ref{modelsmm}, \ref{modelindep}, \ref{model6} and \ref{modelspp} are ordered along the inequalities in \eqref{eq:InegalitesSSSS}. However, Models \ref{modelsmm} and \ref{modelindepm} are not comparable under the supermodular order; we cannot compare $\Smm$ and $\Sindepm$ under the increasing convex order. 
    As mentioned in \cite[Section 7.2]{blier2021stochastic}, we cannot refer to Theorem \ref{thm:Resultats} to compare the aggregate claim amount rvs defined within Models \ref{model4} and \ref{model7}.
    We can apply Theorem \ref{thm:Resultats} to compare the aggregate claim amount rvs defined within Models \ref{modelsmp}, \ref{model6} and \ref{modelspp}; we have
    \begin{equation*}
            \Smp \preceq_{icx} \Sindepp \preceq_{icx} \Spp.
    \end{equation*}
    However, the aggregate claim amount rvs defined within Models \ref{modelsmp}, \ref{model4} and \ref{modelindep} are not comparable. 
    
    
\begin{table}[ht]
	\centering
	\begin{tabular}{cccccccc}
		        Model          &    $S$     & $\theta_{01}$ & $\theta_{12}$ & $\theta_{012}$ & $E[S]$ & $E\left[S^2\right]$ & $TVaR_{0.99}(S)$   \\ \hline
		 \newtag{1}{modelsmp}  &   $\Smp$   &      -1       &       1       &       0        & 724.96 &      650248.05      &     1810.88        \\
		 \newtag{2}{modelsmm}  &   $\Smm$   &     -1/3      &     -1/3      &       0        & 734.99 &      641060.55      &     1690.24        \\
		  \newtag{3}{modelindepm}   & $\Sindepm$ &       0       &      -1       &       0        & 740.00 &      636466.80      &     1585.99        \\
		  \newtag{4}{model4}   &     --     &       0       &       0       &       1        & 740.00 &      655846.68      &     1731.00        \\
		\newtag{5}{modelindep} & $\Sindep$  &       0       &       0       &       0        & 740.00 &      658000.00      &     1742.28        \\
		  \newtag{6}{model6}   & $\Sindepp$ &       0       &       1       &       0        & 740.00 &      679533.20      &     1827.92        \\
		  \newtag{7}{model7}   &     --     &       0       &       0       &       -1       & 740.00 &      660153.32      &     1752.93        \\ 
		 \newtag{8}{modelspp}  &   $\Spp$   &       1       &       1       &       0        & 755.04 &      708818.36      &     1843.25      
	\end{tabular}
	\caption{Results for the aggregate claim amount rv within Models \ref{modelsmp} to \ref{modelspp} in Example \ref{ex:finite-support}.}
	\label{tab:ex-extremal-points}
\end{table}
\end{example}

\section{Discussion}\label{sect:conclusion}

In this paper, we study the effect of dependence within CRMs with full dependence. In particular, we study the aggregate claim amount rv derived from CRMs $(N, \underline{X})\in \aleph^{FGM*}$. We develop convenient representations for the moments and Laplace-Stieltjes transforms of aggregate claim amount rvs. Since the entire class of multivariate symmetric Bernoulli distributions map to the class of FGM copulas, it facilitates the investigation of the effect of dependence on the resulting aggregate claim amount rv within a CRM. Through examples, we show that, even if a FGM copula admits moderate dependence, it significantly affects the resulting aggregate claim amount rv defined within the corresponding CRM, especially when the dependence structure between the frequency and the claim amounts is negative. 

The expressions for the moments and the LST developed in this paper are more convenient than the results previously explored in the literature of copula-based CRMs. Within the previous approach, the complexity arises from the formation of the CRM based on \eqref{eq:cdf-nx1x2x3}, representing the cumulative distribution function of a mixed random vector, where $N$ serves as a frequency rv and the sequence $\{X_j\}_{\{j \geq 1\}}$ can be discrete, continuous, or mixed. Generally, to obtain formulas for moments or LSTs, it is essential to employ expressions involving the conditional cdf $F_{X_1,\dots, X_k \vert N=n}$ and the conditional density function $f_{X_1,\dots, X_k \vert N=n}$, applicable for $k\in \mathbb{N}_1$ and $n \in A_N$. These expressions tend to be cumbersome; hence, deriving moments or LSTs that yield simple expressions is uncommon. Within this paper, we circumvent the difficulties arising from copulas with discrete rvs by using the stochastic representation proposed in Theorem \ref{thm:stochastic-representation-s}. 
We have built on this stochastic representation to analyze the impact of dependence on the aggregate claim amount rv $S$. We have found that within the class $\aleph^{FGM}$, the most dangerous case under the increasing convex order is when all components of the symmetric Bernoulli random vector are comonotonic, meaning that the dependence structure of $(N, \underline{X})$ is the extreme positive dependence FGM copula. The impact of negative dependence has shown to be more difficult to infer. Caution is required to find the safest CRM when the number of claims exhibits negative dependence on the individual claim amounts. We have been able to order different CRMs under the increasing convex order, but no general statement under negative dependence.

Further work includes expanding the results from this paper to generalized FGM copulas and Bernstein copulas. Our results extend to constructing multivariate versions of copula-based CRMs, leading to multi-period or multi-peril CRMs; we defer this analysis to other work. 

\section*{Funding}
This work was partially supported by the Natural Sciences and Engineering Research Council of Canada (Cossette: 04273; Marceau: 05605). The authors report there are no competing interests to declare. We thank the handling editor and the anonymous referees for their insightful comments. 

\bibliographystyle{apalike}
\bibliography{ref}

\appendix

\section{Background on stochastic orders}\label{app:defn-stochastic-orders}

In this paper, we rely on stochastic orders to study the effect of dependence on different CRMs and aggregate claim amount rvs within CRMs. Throughout the paper, we use notions of stochastic orders, assuming the reader is familiar with these notions. 
We define the relevant stochastic orders to make the paper self-contained. 

\begin{definition}
	Let $Z$ and $Z^{\dagger}$ be two rvs with cdfs $F_Z$ and $F_{Z'}$, respectively. We say that $Z$ is smaller than $Z^{\dagger}$ under the usual stochastic order, 
    denoted $Z \preceq_{st} Z^{\dagger}$, 
    if $F_Z(x) \geq F_{Z'}(x)$, for all $x \in \mathbb{R}$. 
    Equivalently, $Z \preceq_{st} Z^{\dagger}$,
    if $E[\phi(Z)] \leq E[\phi(Z^{\dagger})]$ for every increasing function $\phi$, when the expectations exist.
\end{definition}

\begin{definition}
	Let $Z$ and $Z^{\dagger}$ be two rvs with finite expectations. We say that $Z$ is smaller than $Z^{\dagger}$ under the (increasing) convex order, 
    denoted as $Z (\preceq_{icx}) \preceq_{cx} Z^{\dagger}$, 
    if $E[\phi(Z)] \leq E[\phi(Z^{\dagger})]$ for every (increasing) convex function $\phi$, when the expectations exist.
\end{definition}

In actuarial science, the increasing convex order is typically called the stop loss order. Among the relevant implications of the relation 
$Z \preceq_{icx} Z^{\dagger}$, we mention the inequalities $E[Z] \leq  E[Z^{\dagger}]$ and $\text{TVaR}_{\kappa}(Z) \leq \text{TVaR}_{\kappa}(Z^{\dagger})$, for all $\kappa \in (0, 1)$. 
A more detailed list of the implications of the increasing convex order is provided in \cite{muller2002comparison,denuit2006actuarial, shaked2007stochastic}. 
Moreover, if $Z \preceq_{icx} Z^{\dagger}$ 
and $E[Z] = E[Z^{\dagger}]$, we say that $Z$ is smaller than $Z^{\dagger}$ under the convex order, denoted as $Z \preceq_{cx} Z^{\dagger}$.

\begin{definition}
    Let $Z$ and $Z^{\dagger}$ be two rvs with cdfs $F_Z$ and $F_{Z^{\dagger}}$, respectively.
    If 
    $F_Z^{-1}(u) - F_Z^{-1}(v) \leq F_{Z^{\dagger}}^{-1}(u) - F_{Z^{\dagger}}^{-1}(v)$,
    for $0<v \leq u<1$,
    then $Z$ is said to be smaller than $Z^{\dagger}$ according to the dispersive order, which we denote 
    $X \preceq_{disp} Z^{\dagger}$.
\end{definition}

To qualify the impact of dependence within the same family $\aleph^{FGM}$ of CRMs with FGM dependence, we use the supermodular order as presented in Sections 3.8 and 3.9 of \cite{muller2002comparison}.
\begin{definition}[Supermodular order] \label{defSupermodularOrder}
	We say 
    $(V_0, V_1,\dots,V_k)$ is smaller 
    than $(V_0^{\dagger}, V_1^{\dagger},\dots,V_k^{\dagger})$ under the supermodular order, 
    denoted $(V_0, V_1,\dots,V_k) \preceq_{sm}(V_0^{\dagger}, V_1^{\dagger},\dots,V_k^{\dagger})$, if $E\left[\phi (V_0, V_1,\dots,V_k)\right] \leq E\left[\phi (V_0^{\dagger}, V_1^{\dagger},\dots,V_k^{\dagger})\right]$ for all supermodular functions $\phi $, given that the expectations exist. 
    A function $\phi :\mathbb{R}^{k+1}\rightarrow \mathbb{R}$ is said to be supermodular if
	\begin{align*}
		\phi (x_0, x_{1},&\dots,x_{i}+\varepsilon ,\dots,x_{j}+\delta ,\dots,x_{k})-\phi
		(x_0,x_{1},\dots,x_{i}+\varepsilon ,\dots,x_{j},\dots,x_{k}) \\
		& \quad \geq \phi (x_0,x_{1},\dots,x_{i},\dots,x_{j}+\delta ,\dots,x_{k})-\phi
		(x_0,x_{1},\dots,x_{i},\dots,x_{j},\dots,x_{k})
	\end{align*}
	holds for all $(x_0,x_1, \dots, x_k)\in \mathbb{R}^{k+1}$, $0\leq
	i < j\leq k$\ and all $\varepsilon$, $\delta >0$.
\end{definition}
As mentioned in Section 3.8 of \cite{muller2002comparison}, the supermodular order satisfies the nine desired properties for dependence orders. One finds additional details on the supermodular order and its application in  \cite{shaked2007stochastic} and \cite{denuit2006actuarial}. 

Establishing the supermodular order between two vectors of rvs allows one to establish the increasing convex order between the sums of the components of those two vectors, as stated in the following proposition, introduced in \cite[Theorem 3.1]{muller1997stop}); the proof is also provided in \cite[Theorem 8.3.3]{muller2002comparison} or \cite[Proposition 6.3.9]{denuit2006actuarial}. 
\begin{proposition}\label{prop:order-cx-s}
	If $(X_1, \dots X_k) \preceq_{sm} (X_1^{\dagger}, \dots X_k^{\dagger})$ holds, then $\sum_{j=1}^k X_j \preceq_{icx} \sum_{j=1}^k X_j^{\dagger}$. 
\end{proposition}


\section{On Wald's equation}\label{app:wald}

From \eqref{eq:expected-value-s-natural} in Theorem \ref{thm:ES}, one can see that Wald's equation (see \cite{wald1945some}) in \eqref{eq:esperance-indep} holds within the context of this paper when we have $\theta_{0j} = 0$, for $j \in \mathbb{N}_1$. Indeed, Wald's equation holds under the assumption that 
\begin{equation}\label{eq:suff-cond-wald}
	E\left[X_j \times \mathbbm{1}_{\{N \geq n\}}\right] = E[X_j] \times \Pr(N \geq n)
\end{equation}
for all $j \in \mathbb{N}_1$ and $n \in A_N$; this is indeed the case for CRMs in $\aleph^{FGM}$ with $\theta_{0j} = 0$, for $j \in \mathbb{N}_1$. 

Let us consider a count rv $N$ with support of $\{0, 1, 2\}$ and a CRM $(N, X_1, X_2) \in \aleph^{FGM}$. The FGM copula associated with that CRM has four parameters $\theta_{01}, \theta_{02}, \theta_{12}$ and $\theta_{012}$. The conditions of Wald's equation hold if $\theta_{01} = \theta_{02} = 0$. In that case, we have that 
$$F_{N,X_1,X_2}(k,x_1,x_2) = C(F_N(k),F_{X_1}(x_1),F_{X_2}(x_2))$$
with
\begin{equation}\label{eq:copula-wald}
	C(u_0, u_1, u_2) = u_0u_1u_2\left(\theta_{12}\overline{u}_1 \overline{u}_2 + \theta_{012}\overline{u}_0\overline{u}_1\overline{u}_2\right),
\end{equation}
The admissible domain for the parameters in \eqref{eq:copula-wald} is represented by the region $\{(\theta_{12}, \theta_{012})\in [-1, 1]^2: |\theta_{12}| + |\theta_{012}| \leq 1\}$ and corresponds to a convex hull with extremal points $(1, 0), (0, 1), (-1, 0)$ and $(0, -1)$. Note that for any copula in \eqref{eq:copula-wald}, the expected value of the aggregate claim amount rv $S$ within the CRM is $E\left[S\right]=E\left[\Sindep\right] = E[N]E[X]$. However, the components of the CRM are not independent. Indeed, we have
\begin{align*}
    Var(S) &= E[N]Var(X) + \frac{\theta_{12}}{4}\left(\mu_{N}^{(2)} - \mu_{N}\right)\left(\mu_{X_{[2]}}-\mu_{X_{[1]}}\right)^2\\
	&\qquad\qquad -\frac{\theta_{012}}{8}\left(\mu_{N_{[2]}}^{(2)} - \mu_{N_{[2]}} - \mu_{N_{[1]}}^{(2)} + \mu_{N_{[1]}}\right)\left(\mu_{X_{[2]}}-\mu_{X_{[1]}}\right)^2\\
 &\qquad\qquad + Var(N)E[X]^2,
\end{align*}
where the first term comes from $C_{EVar}^{Var}(S)$, the second and third terms come from $C_{ECov}^{Var}(S)$ and the final term comes from the component $C_{VarE}^{Var}(S)$. 

\begin{example}\label{ex:wald}
    Let $N$ be a discrete rv defined by the pmf $\gamma_N(0) = 1/16, \gamma_N(1) = 3/8, \gamma_N(2) = 9/16$ and $\gamma_N(n) = 0$ for $n \notin \mathbb{N}_0 \setminus \{0, 1, 2\}$. Further, let $X_1 \,{\buildrel \mathrm{d} \over =}\,X_2 \eqd \mathrm{Gamma}(5, 3/8)$. Here, whenever $|\theta_{12}| + |\theta_{012}| \leq 1$, we have $E[S] = 20$, $C_{EVar}^{Var}(S) = 160/3$ and $C_{VarE}^{Var}(S) = 200/3$. We present the effect of the variance for different dependence structures in Table \ref{tab:ex-wald}. 
\begin{table}[ht]
	\centering
	\begin{tabular}{rrcrr}
		$\theta_{12}$ & $\theta_{012}$ &        $S$ & $C_{ECov}^{Var}(S)$ & $Var(S)$ \\ \hline
		            0 &              0 &  $\Sindep$ &                               0 &   120.00 \\
		            1 &              0 & $\Sindepp$ &                           12.11 &   132.11 \\
		            -1 &              0 & $\Sindepm$ &                          -12.11 &   107.89 \\
		            0 &              1 &         -- &                           -5.30 &   114.70 \\
		            0 &             -1 &         -- &                            5.30 &   125.30
	\end{tabular}
	\caption{Effect of $\theta_{12}$ and $\theta_{123}$ on $Var(S)$ when $E[S] = E\left[\Sindep\right]$ in Example \ref{ex:wald}.}\label{tab:ex-wald}
\end{table}

\end{example}

\section{Decomposition of the variance}

This section contains an example that illustrates the impact of the claim amount rv, the claim counts rv, and the dependence structure on the three components of the variance in Theorem \ref{thm:var-s}.

\begin{example}
\label{ex:MegaExemple}
Let us first introduce a CRM using Theorem \ref{thm:stochastic-representation-s} where $\{I_j\}_{j \geq 1}$ is a sequence of comonotonic rvs and where $I_0$ is independent of the sequence $\{I_j\}_{j \geq 1}$. We will denote the aggregate claim amount rv within this CRM as $\Sindepp$. 

Within this example, $\mathrm{Ga}$ and $\mathrm{Pa}$ refer respectively to $\mathrm{Gamma}(2, 1/1000)$ and $\mathrm{Pareto}(2.1, 2200)$ such that both severity distributions have mean 2000. Further, $\mathrm{Po}$ refers to a Poisson distribution with rate $\mu_N$, while $\mathrm{NB}$ refers to a negative binomial distribution with parameter $r = 2$ and success probability $1/(1 + \mu_N / r)$ such that the expected frequency is $\mu_N$. Within this example, we will let $\mu_N$ be either 2 or 100, such that the success probability is either 1/2 or 1/51, respectively. 

In Table \ref{tab:decomposition-variance}, we decompose the variance of $S$ into its three components. In general, we observe that we obtain a high coefficient of variation when $\mu_N = 2$ and when the claim amount rv has a higher variance (with the Pareto claim amounts); this is attributable to the increased impact of dependence due to the rvs $(X \vert N = 0)$ being unobserved. Also, remark that $C_{ECov}^{Var}(S)$ is the same for $\Smp$ and $\Spp$. For $\Sindep$ and $\Sindepp$, the components $C_{EVar}^{Var}(S)$ and $C_{VarE}^{Var}(S)$ are the same, but the higher value of $Var(\Sindepp)$ compared with $Var(\Sindep)$ is caused by the second term which captures the effect of dependence between claim amount rvs. 

\begin{table}[ht]
	\centering
	\resizebox{\textwidth}{!}{
		\begin{tabular}{rrrrrrrrrr}
			                 Dependence & $X \eqd$ & $N \eqd$ & $\mu_{N}$ & $E[S]$ &    $Var(S)$ & $\frac{\sqrt{Var(S)}}{E[S]}$ & $C_{EVar}^{Var}(S)$ & $C_{ECov}^{Var}(S)$ & $C_{VarE}^{Var}(S)$ \\ \hline
			    \multirow{8}{*}{$\Smp$} &       Ga &       Po &         2 &   3421 &     7465515 &                         0.80 &             3023803 &             1016498 &             3425214 \\
			                            &       Ga &       NB &         2 &   3222 &    10881173 &                         1.02 &             2748347 &             1104064 &             7028761 \\
			                            &       Pa &       Po &         2 &   2987 &   103531039 &                         3.41 &            98722551 &             3113024 &             1695465 \\
			                            &       Pa &       NB &         2 &   2639 &    84307967 &                         3.48 &            77817855 &             3381196 &             3108915 \\
			                            &       Ga &       Po &       100 & 195771 &  4502910641 &                         0.34 &           177003907 &          3724182732 &           601724003 \\
			                            &       Ga &       NB &       100 & 171596 & 11102653630 &                         0.61 &           149514222 &          3849370885 &          7103768524 \\
			                            &       Pa &       Po &       100 & 192600 & 22157902045 &                         0.77 &          7503591483 &         11405309616 &          3249000947 \\
			                            &       Pa &       NB &       100 & 150292 & 18896345454 &                         0.91 &          4984088480 &         11788698335 &          2123558639 \\ \hline
                 \multirow{8}{*}{$\Sindep$} &       Ga &       Po &         2 &   4000 &    12000000 &                         0.87 &             4000000 &                   0 &             8000000 \\
			                            &       Ga &       NB &         2 &   4000 &    20000000 &                         1.12 &             4000000 &                   0 &            16000000 \\
			                            &       Pa &       Po &         2 &   4000 &   168000000 &                         3.24 &           160000000 &                   0 &             8000000 \\
			                            &       Pa &       NB &         2 &   4000 &   176000000 &                         3.32 &           160000000 &                   0 &            16000000 \\
			                            &       Ga &       Po &       100 & 200000 &   600000000 &                         0.12 &           200000000 &                   0 &           400000000 \\
			                            &       Ga &       NB &       100 & 200000 & 20600000000 &                         0.72 &           200000000 &                   0 &         20400000000 \\
			                            &       Pa &       Po &       100 & 200000 &  8400000000 &                         0.46 &          8000000000 &                   0 &           400000000 \\
			                            &       Pa &       NB &       100 & 200000 & 28400000000 &                         0.84 &          8000000000 &                   0 &         20400000000 \\ \hline
			\multirow{8}{*}{$\Sindepp$} &       Ga &       Po &         2 &   4000 &    14250000 &                         0.94 &             4000000 &             2250000 &             8000000 \\
			                            &       Ga &       NB &         2 &   4000 &    23375000 &                         1.21 &             4000000 &             3375000 &            16000000 \\
			                            &       Pa &       Po &         2 &   4000 &   174890625 &                         3.31 &           160000000 &             6890625 &             8000000 \\
			                            &       Pa &       NB &         2 &   4000 &   186335938 &                         3.41 &           160000000 &            10335938 &            16000000 \\
			                            &       Ga &       Po &       100 & 200000 &  6225000000 &                         0.39 &           200000000 &          5625000000 &           400000000 \\
			                            &       Ga &       NB &       100 & 200000 & 29037500000 &                         0.85 &           200000000 &          8437500000 &         20400000000 \\
			                            &       Pa &       Po &       100 & 200000 & 25626562500 &                         0.80 &          8000000000 &         17226562500 &           400000000 \\
			                            &       Pa &       NB &       100 & 200000 & 54239843750 &                         1.16 &          8000000000 &         25839843750 &         20400000000 \\ \hline
			    \multirow{8}{*}{$\Spp$} &       Ga &       Po &         2 &   4579 &    20364862 &                         0.99 &             4181061 &             1016498 &            15167304 \\
			                            &       Ga &       NB &         2 &   4778 &    34658951 &                         1.23 &             4303903 &             1104064 &            29250984 \\
			                            &       Pa &       Po &         2 &   5013 &   244199489 &                         3.12 &           218842344 &             3113024 &            22244121 \\
			                            &       Pa &       NB &         2 &   5361 &   284658661 &                         3.15 &           239279661 &             3381196 &            41997804 \\
			                            &       Ga &       Po &       100 & 204229 &  7911324287 &                         0.44 &           185461456 &          3724182732 &          4001680099 \\
			                            &       Ga &       NB &       100 & 228404 & 45358727233 &                         0.93 &           206323009 &          3849370885 &         41303033339 \\
			                            &       Pa &       Po &       100 & 207400 & 28985692423 &                         0.82 &          8381458691 &         11405309616 &          9198924116 \\
			                            &       Pa &       NB &       100 & 249708 & 84641633438 &                         1.17 &         10880663037 &         11788698335 &         61972272066
		\end{tabular}
	}
	\caption{Expected value, variance, and three variance components for different CRMs. }
	\label{tab:decomposition-variance}
\end{table}

In the component $C_{EVar}^{Var}(S)=E[NVar(X\vert N)]$, it is the variance of $X$ that dominates, so the CRMs constructed with the Pareto distributions have higher values of $C_{EVar}^{Var}(S)$. On the other hand, the component $C_{VarE}^{Var}(S) = Var(NE[X\vert N])$ is dominated by the variance of $N$, so the third component is larger with $\mathrm{NB}$ rather than with $\mathrm{Po}$. The second component $C_{ECov}^{Var}(S)$ depends both on the expected spacing squared for the claim amount rv $X$ and the second factorial moment of $N$. 

It is interesting to note that there is not a component of the variance which dominates the others: the relative size of each component depends on a combination of the frequency rv $N$, the claim amount rv $X$ and the dependence structure. In particular,
\begin{itemize}
    \item The line which maximizes the ratio $C_{EVar}^{Var}(S)/Var(S)$ is $\Smp$ with $\mathrm{Po}(2)$ and $\mathrm{Pa}$; the proportion of $C_{EVar}^{Var}(S)$, $C_{ECov}^{Var}(S)$, $C_{VarE}^{Var}(S)$ are respectively 0.95, 0.03, 0.02.
    \item The line which maximizes the ratio $C_{ECov}^{Var}(S)/Var(S)$ is $\Sindepp$ with $\mathrm{Po}(100)$ and $\mathrm{Ga}$; the proportion of $C_{EVar}^{Var}(S)$, $C_{ECov}^{Var}(S)$, $C_{VarE}^{Var}(S)$ are respectively 0.03, 0.9, 0.06.
    \item The line which maximizes the ratio $C_{VarE}^{Var}(S)/Var(S)$ is $\Sindep$ with $\mathrm{NB}(100)$ and $\mathrm{Ga}$; the proportion of $C_{EVar}^{Var}(S)$, $C_{ECov}^{Var}(S)$, $C_{VarE}^{Var}(S)$ are respectively 0.01, 0, 0.99.    
\end{itemize}

\end{example}

\section{Additonal applications of Theorem \ref{thm:laplace-crm-identical}}
\label{sect:AdditionalApplicationsTheoremLaplace}

In Theorem \ref{thm:laplace-crm-identical}, when $N$ is independent of the sequence of claim amount rvs, the entire dependence structure is contained within the rv $K_n$ for $n\in \mathbb{N}_1$ and \eqref{eq:lse-with-k} becomes, for $t>0$, 
$$\mathcal{L}_{S}(t) = \gamma_{N}(0) + \sum_{n = 1}^{\infty}\sum_{k = 0}^{n}\gamma_N(n)\Pr(K_n = k)\mathcal{L}_{X_{[1]}}(t)^{n-k}\mathcal{L}_{X_{[2]}}(t)^{k} = E_{N}\left[\mathcal{L}_{X_{[1]}}(t)^N \mathcal{P}_{K_{N}}\left(\frac{\mathcal{L}_{X_{[2]}}(t)}{\mathcal{L}_{X_{[1]}}(t)}\right)\right].$$



In some cases, one may use Corollary \ref{cor:lst-min-max} to obtain closed-form expressions for LSTs, as the following example shows. 
\begin{example} \label{ex:exponentielle} 
Let $N \eqd \mathrm{Geometric}(p)$ and $\underline{X}$ form a sequence of identically distributed exponential rvs with mean $1/\beta$, for $\beta > 0$. Then, the LSTs of the aggregate claim amount rvs defined within the CRMs in Corollary \ref{cor:lst-min-max} become
\begin{align*}
	\mathcal{L}_{S^{(\bigtriangleup, \bigtriangleup)}}(t) &= p + (1-p) \left(
	\frac{1-p}{2}\mathcal{L}_{Y_2}(t) + \frac{\beta_4}{\beta_4 - \beta_3}\mathcal{L}_{Y_3}(t) + \frac{\beta_3}{\beta_3 - \beta_4}\mathcal{L}_{Y_4}(t) - \right.\\
	& \qquad \qquad \left.
	\frac{1-p}{2}\frac{\beta_6}{\beta_6 - \beta_5} \mathcal{L}_{Y_5}(t) -
	\frac{1-p}{2}\frac{\beta_5}{\beta_5 - \beta_6} \mathcal{L}_{Y_6}(t)
	\right)
\end{align*}
and
\begin{align*}
	\mathcal{L}_{S^{(\bigtriangledown, \bigtriangleup)}}(t) &= p + (1-p) \left(
	-\frac{1-p}{2} \mathcal{L}_{Y_1}(t) + \mathcal{L}_{Y_2}(t) + \right.\\
	& \qquad \qquad \left.
	\frac{1-p}{2}\frac{\beta_6}{\beta_6 - \beta_5} \mathcal{L}_{Y_5}(t) + \frac{1-p}{2}\frac{\beta_5}{\beta_5 - \beta_6} \mathcal{L}_{Y_6}(t)
	\right),
\end{align*}
for $t \geq 0$, where $\mathcal{L}_{Y_j}(t) = \beta_j/(\beta_j + t)$ for $j \in \{1, \dots, 6\}$, with 
$\beta_1 = 2\beta p$, 
$\beta_2 = 2\beta p(2-p)$, 
$\beta_3 = \frac{\beta}{2}(3 - \sqrt{9 - 8p})$, 
$\beta_4 = \frac{\beta}{2}(3 + \sqrt{9 - 8p})$, 
$\beta_5 = \frac{\beta}{2}(3 - \sqrt{1 + 8(1-p)^2})$ and 
$\beta_6 = \frac{\beta}{2}(3 + \sqrt{1 + 8(1-p)^2})$. It follows that the LST of $\Spp$ and $\Smp$ lead to closed-form expressions for the cdfs. 
Note that 
$E[S^{(\bigtriangledown, \bigtriangleup)}] 
= E[N]E[X]\left(1 - 1/\{2(2 - p)\}\right)$ 
and 
$E[S^{(\bigtriangleup, \bigtriangleup)}] 
= E[N]E[X]\left(1 + 1/\{2(2 - p)\}\right) $.

Assume further that $p = 10/11$ and $\beta = 1/2000$, such that $E[N] = 0.1$ and $E[X] = 2000$ as in Example \ref{ex:pareto}. 
In Table \ref{tab:exponentielle-mean-variance}, we provide the values of the expectations and the variances of the aggregate claim amount rvs. 
It is interesting to compare these values with the results of Example \ref{ex:pareto} in which the claim amount rvs are Pareto distributed with identical mean as the exponential here. As expected, both examples provide the same results for $S^{(\perp, \perp)}$. For $\Smp$, the mean is lower in Example \ref{ex:pareto}, but the variance is much higher here. The reverse order is observed for $\Spp$, with the Pareto distribution contributing to a huge variance compared with the exponential distribution.
	\begin{table}[ht]
		\centering
		\begin{tabular}{crr}
			& Mean     & Variance     \\ \hline
			$\Smp$ & 108.33  & 258~819.4   \\
			$S^{(\perp, \perp)}$                 & 200.00 & 840~000.0  \\
            $\Spp$  & 281.67 & 1~444~375.0
		\end{tabular}
		\caption{Mean and variance of aggregate claim amount rvs defined within some CRMs with full FGM dependence in Example \ref{ex:exponentielle}.}\label{tab:exponentielle-mean-variance}
	\end{table}
\end{example}

Another alternate representation of the LST may be useful to construct classes of CRMs within $\aleph^{FGM*}$. Conditioning on $I_0$, we have
\begin{align}
	\mathcal{L}_{S}(t) &= \frac{1}{2} \left\{\gamma_{N_{[1]}}(0) + \sum_{n = 1}^{\infty}\gamma_{N_{[1]}}(n) \sum_{k = 0}^{n}  \Pr(K_n = k \vert I_0 = 0) \mathcal{L}_{X_{[1]}}(t)^{n-k} \mathcal{L}_{X_{[2]}}(t)^{k}\right\} \nonumber\\
  &\qquad + \frac{1}{2} \left\{\gamma_{N_{[2]}}(0) + \sum_{n = 1}^{\infty}\gamma_{N_{[2]}}(n) \sum_{k = 0}^{n}  \Pr(K_n = k \vert I_0 = 0) \mathcal{L}_{X_{[1]}}(t)^{n-k} \mathcal{L}_{X_{[2]}}(t)^{k}\right\}, \quad t > 0.\label{eq:lse-with-k-v2}
\end{align}
The expression in \eqref{eq:lse-with-k-v2} offers a nice interpretation of the distribution of $S$. 
Indeed, the distribution of $S$ can be represented as a mixture of the distribution of two random variables, say $Y_0$ and $Y_1$, 
where the LST of $Y_j$ is $\gamma_{N_{[1+j]}}(0) + \sum_{n = 1}^{\infty}\gamma_{N_{[1+j]}}(n) \sum_{k = 0}^{n}  \Pr(K_n = k \vert I_0 = 0) \mathcal{L}_{X_{[1]}}(t)^{n-k} \mathcal{L}_{X_{[2]}}(t)^{k}$, $t \geq 0$.
Then, it means that $Y_j$ admits a representation as a random sum of dependent rvs, which are independent of the counting rv $N_{[j]}$, $j = 1,2$. 

We illustrate the usefulness of the latter representation of the LST to construct a family of aggregate claim amount rvs defined within a CRM in $\aleph^{FGM*}$.

\begin{example}
    Define the random vector $(I_0,K_n)$, for $n \in \mathbb{N}_1$, such that $I_0 \eqd \mathrm{Bern}(1/2)$ and 
    \begin{align*}
        \Pr(K_n = k | I_0 = 0) 
        &=
        \binom{n}{k} (1-\alpha)^k \alpha^{n-k},
        \quad 
        k \in \{0,1,\dots,n\} \\
        \Pr(K_n = k | I_0 = 1) 
        &=
        \binom{n}{k} \alpha^k (1-\alpha)^{n-k},
        \quad 
        k \in \{0,1,\dots,n\},
    \end{align*}
    where $\alpha \in (0,1)$. One may construct a family of aggregate claim amount rvs, denoted by $S^{(\alpha)}$, by replacing $(I_0,K_n)$ within \eqref{eq:lse-with-k-v2}. Then, we have
    \begin{equation*}
        \mathcal{L}_{S^{(\alpha)}}(t) 
        =
        \frac{1}{2} 
        \mathcal{P}_{N_{[1]}}\left(
        \alpha \mathcal{L}_{X_{[1]}}(t) 
        + (1-\alpha) \mathcal{L}_{X_{[2]}}(t) \right)
         + 
        \frac{1}{2} 
        \mathcal{P}_{N_{[2]}}\left(
        (1-\alpha) \mathcal{L}_{X_{[1]}}(t) + \alpha \mathcal{L}_{X_{[2]}}(t) \right),
        \quad 
        t \geq 0.
        \label{eq:LSTJolie}
    \end{equation*}
\end{example}

\end{document}